\documentclass[11pt]{article}

\usepackage{amssymb, color}
\usepackage{amsfonts}
\usepackage{amsmath}
\usepackage{latexsym}
\usepackage{epsfig}
\usepackage{bm}
\usepackage{times}

\newif\ifhyper\IfFileExists{hyperref.sty}{\hypertrue}{\hyperfalse}
\hypertrue
\ifhyper\usepackage{hyperref}\fi

\newcommand{\ignore}[1]{}

\makeatletter
\renewcommand{\subsubsection}{\@startsection{subsubsection}{3}{0pt}{-12pt}{-5pt}{\normalsize\bf}}
\makeatother

\newenvironment{proof}{\noindent \textbf{Proof:}\nopagebreak[2]}{$\qed$}
\newenvironment{proofof}[1]{\noindent \textbf{Proof of {#1}:}\nopagebreak[2]}{$\qed$}

\newcommand{\qed}{\hfill\rule{7pt}{7pt} \medskip}
\newtheorem{claim}{Claim}
\newtheorem{proposition}[claim]{Proposition}

\newtheorem{lemma}[claim]{Lemma}
\newtheorem{theorem}{Theorem}

\newtheorem{fact}[claim]{Fact}

\newcommand{\dtv}{d_{\mathrm TV}}
\newcommand{\dk}{d_{\mathrm K}}

\newcommand{\R}{{\mathbb R}}
\newcommand{\Z}{{\mathbb Z}}

\newcommand{\E}{{\bf E}}

\newcommand{\eps}{\epsilon}
\newcommand{\littlesum}{\mathop{{\textstyle \sum}}}

\newcommand{\poly}{\mathrm{poly}}

\newcommand{\wh}[1]{{\widehat{#1}}}

\newcommand{\opt}{\mathsf{opt}}

\newcommand{\AND}{{\mathrm{L^{\uparrow}}}}
\newcommand{\ANI}{{\mathrm{L^{\downarrow}}}}

\newcommand{\TND}{{\mathrm{T^{\uparrow}}}}
\newcommand{\TNI}{{\mathrm{T^{\downarrow}}}}

\newcommand{\dis}{{\mathcal{D}}}
\newcommand{\mon}{{\mathcal{M}}}
\newcommand{\A}{{\mathcal{A}}}

\newcommand{\norm}[1]{\left\|#1\right\|}

   \newcommand{\rnote}[1]{\footnote{{\bf [[Rocco: {#1}\bf ]] }}}
   \newcommand{\inote}[1]{\footnote{{\bf [[Ilias: {#1}\bf ]] }}}
   
\newcommand{\new}[1]{{{#1}}}
\newcommand{\newer}[1]{{{#1}}}

\begin{document}

\title{Learning $k$-Modal Distributions via Testing\footnote{A preliminary version of this work appeared in the {\em Proceedings of the 
Twenty-Third Annual ACM-SIAM Symposium on Discrete Algorithms (SODA 2012).}}}

\author{Constantinos Daskalakis\thanks{Research supported by NSF CAREER award CCF-0953960 and by a
Sloan Foundation Fellowship.}
\\
MIT\\
{\tt costis@csail.mit.edu}\\
\and
Ilias Diakonikolas\thanks{Most of this research was done while the author was at UC Berkeley 
supported by a Simons Postdoctoral Fellowship. Some of this work was done at Columbia University, supported by NSF grant CCF-0728736, and by an Alexander S. Onassis Foundation
Fellowship.}\\
University of Edinburgh\\
{\tt ilias.d@ed.ac.uk}\\
\and
Rocco A. Servedio
\thanks{Supported by NSF grants CNS-0716245, CCF-0915929, and CCF-1115703.}\\
Columbia University\\
{\tt rocco@cs.columbia.edu}}

\setcounter{page}{0}

\maketitle

\thispagestyle{empty}

\begin{abstract}
A $k$-modal probability distribution over the 
\new{discrete} domain $\{1,...,n\}$ is one whose
histogram has at most $k$ ``peaks'' and ``valleys.'' Such distributions are
natural generalizations of monotone ($k=0$) and unimodal ($k=1$)  probability distributions, which  have been intensively studied in probability theory and statistics.

In this paper we consider the problem of \emph{learning} 
\new{(i.e.,performing density estimation of)}
an unknown $k$-modal distribution with respect to the $L_1$ distance.
The learning algorithm is given access to independent samples 
drawn from an unknown $k$-modal distribution $p$, and 
it must output a hypothesis distribution $\widehat{p}$ such that with high
probability the total variation distance between $p$ and $\widehat{p}$ 
is at most $\eps.$  \new{Our main goal is to obtain
\emph{computationally efficient} algorithms for this problem that
use (close to) an information-theoretically optimal number of samples.
}

We give an efficient algorithm for this problem that runs in time $\poly(k,\log(n),1/\eps)$.
For $k \leq \tilde{O}(\new{\log n})$, the number of samples used by our algorithm is very close (within an
$\tilde{O}(\log(1/\eps))$ factor) to being information-theoretically optimal.  Prior to this
work computationally efficient algorithms were known only for the cases $k=0,1$
\cite{Birge:87b,Birge:97}.

A novel feature of our approach is that our learning algorithm crucially 
uses a new algorithm for \emph{property testing of probability distributions} 
as a key subroutine.  
The learning algorithm uses the property tester to efficiently
decompose the $k$-modal distribution into $k$ (near-)monotone distributions, 
which are easier to
learn.

\end{abstract}



\section{Introduction} \label{sec:intro}


This paper considers a natural unsupervised learning problem involving \emph{$k$-modal} distributions
  over the discrete domain $\new{[n]=}\{1,\dots,n\}.$ A distribution is
$k$-modal if the plot of its probability density function (pdf) has at most $k$ ``peaks''
and ``valleys'' (see Section~\ref{ssec:defs} for a precise definition).  Such distributions arise both in theoretical (see e.g., \cite{CKC:83,Kemperman:91,ChanTong:04}) and applied (see e.g., \cite{Murphy:64,ToledoFernandez:90,FPPRD:98}) research; they naturally generalize the simpler classes of monotone ($k=0$) and unimodal ($k=1$) distributions
that have been intensively studied in probability theory and statistics (see the discussion of related work below).

Our main aim in this paper is to give an efficient algorithm for \emph{learning}
an unknown $k$-modal distribution $p$ to total variation distance $\eps$, given access only to independent samples drawn from $p$.  As described below there is an information-theoretic lower bound of
$\Omega(k \log(n/k)/\eps^3)$ samples for this learning problem, so an important goal for us is to
 obtain an algorithm whose sample complexity is as close as possible to this lower bound.  An equally important goal is for our
algorithm to be computationally efficient, i.e., to run in time 
polynomial in the size of its input sample.  
Our main contribution in this paper is a computationally efficient
algorithm that has nearly optimal sample complexity for 
small (but super-constant) values of $k.$

\subsection{Background and \new{relation to previous work}} \label{sec:background}


There is a rich body of work in the statistics and probability literatures 
on estimating distributions under various kinds of ``shape'' or ``order'' 
restrictions.  In particular, many researchers have studied the risk of 
different estimators for monotone \new{($k=0$)} 
and unimodal \new{($k=1$)} distributions; see for 
example the works of \cite{PrakasaRao:69,Wegman:70,Groeneboom:85,Birge:87, 
Birge:87b, Birge:97}, among many others.
These and related papers from the probability/statistics literature 
mostly deal with information-theoretic
upper and lower bounds on the sample complexity of learning 
monotone and unimodal distributions.  
\new{In contrast, a central goal of the current work is to 
obtain \emph{computationally efficient} learning algorithms for
larger values of $k$.} 

It should be noted that some of the works cited above 
do give efficient algorithms for the cases $k=0$ and $k=1$; 
in particular we mention the results of Birg\'{e} \cite{Birge:87b,Birge:97}, 
which give computationally efficient $O(\log(n)/\eps^3)$-sample algorithms 
for learning unknown monotone or unimodal distributions over $[n]$ 
respectively.
(Birg\'{e}~\cite{Birge:87} also showed that this sample
complexity is asymptotically optimal, as we discuss below; we describe 
the algorithm of \cite{Birge:87b} in more detail
in Section~\ref{ssec:tools}, and indeed use it as an ingredient of our 
approach throughout this paper.)  
However, for these relatively simple $k=0,1$ classes of distributions the 
main challenge is in developing sample-efficient estimators, and the 
algorithmic aspects are typically rather straightforward (as is the case 
in \cite{Birge:87b}).  In contrast, much more challenging and interesting 
algorithmic issues arise for the general values of $k$ which we consider here.

\subsection{Our Results}

Our main result is a highly efficient algorithm for learning an unknown
$k$-modal distribution over $[n]$:

\begin{theorem} \label{thm:main}
Let $p$ be any unknown $k$-modal distribution over $[n]$.  There is an algorithm that uses\footnote{We
write $\tilde{O}(\cdot)$ to hide factors which are poly-logarithmic
in the argument to $\tilde{O}(\cdot)$;
thus for example $\tilde{O}(a \log b)$ denotes a quantity
which is $O((a \log b) \cdot  (\log (a \log b))^c)$ for some
absolute constant $c$.}
\[
\left ( {\frac {k\log(n/k)}{\eps^3}}  +   \frac{\new{k^2}}{\eps^3} \cdot \log {\frac k \eps} \cdot
\log\log {\frac k \eps} \right )
\cdot \tilde{O}(\log(1/\delta))
\]
samples from $p$, runs for $\poly(k, \log n, 1/\eps,\log(1/\delta))$
bit operations, and with probability $1-\delta$ outputs a (succinct description of a) hypothesis
distribution $\widehat{p}$ over $[n]$ such that the total variation distance between $p$ and $\widehat{p}$ is
at most $\eps.$
\end{theorem}


As alluded to earlier, Birg\'{e}~\cite{Birge:87} gave a sample complexity lower bound
for learning monotone distributions. The lower bound in~\cite{Birge:87} is stated for continuous
distributions but the arguments are easily adapted to the discrete case; ~\cite{Birge:87} shows that
(for $\eps \geq 1/n^{\Omega(1)}$)\footnote{For $\eps$ sufficiently small the generic
upper bound of \new{Fact~\ref{thm:folklore}}, which says that any distribution over $[n]$
can be learned to variation distance $\eps$ using $O(n/\eps^2)$ samples, provides a
better bound.} any algorithm for learning an unknown monotone distribution over $[n]$ to
total variation distance $\eps$ must use $\Omega(\log(n)/\eps^3)$ samples.
By a simple construction which concatenates $k$ copies of the monotone lower bound construction over intervals
of length $n/k$, using the monotone lower bound it is possible to show:


\begin{proposition} \label{obs:lowerbound}
Any algorithm for learning an unknown $k$-modal distribution over $[n]$ to variation distance
$\eps$ (for $\eps \geq 1/n^{\Omega(1)}$) must use $\Omega(k \log(n/k)/\eps^3)$ samples.
\end{proposition}


Thus our learning algorithm is nearly optimal in its sample complexity; more precisely,
for $k \leq \tilde{O}(\new{\log n})$ (and $\eps$ as bounded above), our sample complexity in Theorem~\ref{thm:main} is asymptotically optimal up to a factor of $\tilde{O}(\log(1/\eps))$.  Since each draw
from a distribution over $[n]$ is a $\log(n)$-bit string, Proposition~\ref{obs:lowerbound}
implies that the running time of our algorithm is optimal up to polynomial factors.  
\new{As far as we are aware,} prior to this work no learning
algorithm for $k$-modal distributions was known that simultaneously had
$\poly(k,\log n)$ sample complexity \new{and even running time $p(n)$ for a 
fixed polynomial $p(n)$ (where the exponent does not depend on
$k$).
}

\subsection{Our Approach}

As mentioned in Section~\ref{sec:background} Birg\'{e} gave a highly efficient algorithm for learning a
\emph{monotone} distribution in \cite{Birge:87b}.  Since a $k$-modal distribution is simply
a concatenation of $k+1$ monotone distributions (first non-increasing, then non-decreasing, then
non-increasing, etc.),
it is natural to try to use Birg\'{e}'s
algorithm as a component of an algorithm for learning $k$-modal distributions, and indeed this is what we do.

The most naive way to use Birg\'{e}'s algorithm would be to guess all possible ${n \choose k}$ locations of the $k$ ``modes'' of $p$.  While such an approach can be shown to have good sample complexity, the resulting $\Omega(n^k)$ running time is grossly inefficient.  A ``moderately naive'' approach, which we analyze in Section~\ref{ssec:kmodal-learn-simple}, is to partition $[n]$ into roughly $k/\eps$ intervals each of weight roughly $\eps/k$, and run Birg\'{e}'s algorithm separately on each such interval.  Since the target distribution is $k$-modal, at most $k$ of the intervals can be non-monotone; Birg\'{e}'s algorithm can be used to obtain an $\eps$-accurate hypothesis on each monotone interval, and even if it fails badly on the (at most) $k$ non-monotone intervals, the resulting total contribution towards the overall error from those failures is at most $O(\eps).$  This approach is much more efficient than the totally naive approach, giving
running time polynomial in $k$, $\log n$, and $1/\eps$, but its sample complexity turns out to be polynomially worse
than the $O(k\log(n)/\eps^3)$ that we are shooting for.  (Roughly speaking, this is because the approach involves running Birg\'{e}'s $O(\log(n)/\eps^3)$-sample algorithm 
$\new{\Omega}(k/\eps)$ times, so it uses at least $k\log(n)/\eps^4$ samples.)

Our main learning result is achieved by augmenting the ``moderately naive'' 
algorithm sketched above with a new \emph{property testing} algorithm.  
\new{Unlike a learning algorithm, a property testing algorithm
for probability distributions need not output a high-accuracy hypothesis;
instead, it has the more modest goal of successfully (with high probability)
distinguishing between probability distributions that have a given property 
of interest, versus distributions that are far (in total variation distance)
from every distribution that has the property.  See \cite{GGR98,
Ron:10FNTTCS, PropertyTestingICS} for broad overviews of property testing.}

We give a property testing algorithm for the following problem:  given
samples \new{from a distribution $p$ over $[n]$ which is promised to be
$k$-modal}, output ``yes'' (with high probability) if $p$ is 
\emph{monotone} and ``no'' (with high probability) if $p$ is
$\eps$-far \new{in total variation distance} 
from every monotone distribution.  
Crucially, our testing algorithm uses $O(\new{k}/\eps^2)$ samples 
\emph{independent of $n$}
for this problem.  Roughly speaking, by using this algorithm $O(k/\eps)$ 
times we are able to identify
$k+1$ intervals that (i) collectively contain almost all of $p$'s mass, 
and (ii) are each (close to) monotone and thus can be handled using 
Birg\'{e}'s algorithm.  Thus the overall sample complexity of our approach is 
(roughly) $\new{O(k^2/\eps^3)}$ (for the $O(k/\eps)$
runs of the tester) plus $O(k\log(n)/\eps^3)$ (for the $k$ runs 
of Birg\'{e}'s algorithm), which gives Theorem~\ref{thm:main} and is very 
close to optimal for $k$ not too large.

\subsection{Discussion}

Our learning algorithm highlights a novel way that property testing algorithms can be useful
for learning.  Much research has been done on understanding the relation between property testing
algorithms and learning algorithms, see e.g., \cite{GGR98,KR00} and the lengthy survey \cite{Ron:08testlearn}.  As Goldreich has noted \cite{Goldreich:11web}, an often-invoked motivation for
property testing is that (inexpensive) testing algorithms can be used as a ``preliminary diagnostic''
to determine whether it is appropriate to run a (more expensive) learning algorithm.  In contrast, in this work we are using property testing rather differently, as an inexpensive way of decomposing a ``complex'' object (a $k$-modal
distribution) which we do not \emph{a priori} know how to learn, into a collection of ``simpler'' objects (monotone or near-monotone distributions) which can be learned using existing techniques.  We are not aware of prior learning algorithms that successfully use property testers in this way;
we believe that this high-level
approach to designing learning algorithms, by using property testers to decompose ``complex'' objects into simpler objects that can be efficiently learned, may find future applications elsewhere.

\section{Preliminaries} \label{sec:prelims}

\subsection{Notation and Problem Statement} \label{ssec:defs}

For $n \in \Z_+$, denote by $[n]$ the set $\{1, \ldots, n\}$;
for $i,j  \in \Z_+$, $i \leq j$, denote by $[i,j]$ the set $\{i, i+1, \ldots, j\}$.
We write $v(i)$ to denote the $i$-th element of vector $v \in \R^n$.
For $v=(v(1),\dots,v(n)) \in \R^n$ denote by $\|v\|_1 = \littlesum_{i=1}^n |v(i)|$ its $L_1$-norm.

We consider discrete probability distributions over $[n]$, which are
functions $p: [n] \to [0,1]$ such that $\littlesum_{i=1}^n p(i)=1$.
For $S \subseteq [n]$ we write $p(S)$ to denote $\littlesum_{i \in S} p(i)$.
 For $S \subseteq [n]$, we write $p_S$ to denote the {\em conditional distribution}
over $S$ that is induced by $p.$ We use the  notation $P$ for the
{\em cumulative distribution function (cdf)} corresponding to $p$, i.e.,
$P: [n] \to [0,1]$ is defined by $P(j) = \littlesum_{i=1}^j p(i)$.

A distribution $p$ over $[n]$ is non-increasing (resp. non-decreasing)
if $p(i+1) \leq p(i)$ (resp. $p(i+1) \geq p(i)$), for all $i \in [n-1]$;
$p$ is \emph{monotone} if it is either non-increasing or non-decreasing.
We call a nonempty interval $I=[a,b] \subseteq [2, n-1]$ a \emph{max-interval of $p$} if
$p(i) = c$ for all $i \in I$ and $\max\{p(a-1), p(b+1)\} < c$; in this case,
we say that the point $a$ is a \emph{left max point} of $p$.
Analogously, a \emph{min-interval of $p$} is an interval
 $I=[a,b] \subseteq [2, n-1]$ with $p(i) = c$ for all $i \in I$ and $\min \{p(a-1), p(b+1)\} > c$;
 the point $a$ is called a \emph{left min point} of $p.$
If $I=[a,b]$ is either a max-interval or a min-interval (it cannot be both)
we say that $I$ is an \emph{extreme-interval} of $p$, and $a$ is called a \emph{left extreme point} of $p.$
Note that any distribution uniquely defines a collection of extreme-intervals (hence, left extreme points).
We say that $p$ is \emph{$k$-modal} if it has at most $k$ extreme-intervals.
\new{We write $\dis_n$ (resp. $\mon^k_n$) to denote the set of all distributions
(resp. $k$-modal distributions) over $[n]$.}

Let $p, q$ be distributions over $[n]$ with corresponding cdfs $P, Q$.
The {\em total variation distance} between $p$ and $q$ is
$\dtv (p,q) := \max_{S \subseteq [n]} |p(S) - q(S)|  = (1/2)\cdot \|p-q\|_1.$
The {\em Kolmogorov distance} between $p$ and $q$ is defined as
$\dk(p,q):= \max_{j \in [n]} \left| P(j) - Q(j) \right|.$ Note that $\dk(p,q) \le \dtv (p,q).$

\new{
We will also need a more general distance measure that captures the above two metrics as special cases.
Fix a family of subsets $\mathcal{A}$ over $[n]$. 
We define the {\em $\mathcal{A}$--distance} between $p$ and $q$ by $\norm {p-q}_{\mathcal{A}}
:= \max_{A\in \mathcal{A}} |p(A) - q(A)|$. (Note that if $\mathcal{A} = 2^{[n]}$, the powerset of $[n]$, then  the 
$\mathcal{A}$--distance is identified with the total variation distance, while when $\mathcal{A} = \{ [1, j], j \in [n]\}$ 
it is identified with the Kolmogorov distance.)
Also recall that the \emph{VC--dimension} of $\mathcal A$ is the maximum size 
of a subset $X\subseteq [n]$ that is shattered by $\mathcal A$ (a set $X$ is shattered by
$\mathcal A$ if for every $Y \subseteq X$ some $A\in\mathcal A$ satisfies $A\cap X = Y$).
}

\medskip

{\noindent \textbf{Learning $k$-modal Distributions}.}
Given independent samples from an unknown
$k$-modal distribution $p \in \mon^k_n$
and $\eps>0$, the goal is to output a hypothesis
distribution $h$ such that with probability
$1-\delta$ we have $\dtv(p,h) \leq \eps$.
We say that such an algorithm $\A$ {\em learns $p$ to accuracy $\eps$
and confidence $\delta$}.  \ignore{
We say that an algorithm is {\em proper}, if in addition we have
that $h \in \mon^k_n$.\rnote{Are we going to get a proper learning algorithm, or should I delete
this sentence?} \inote{Yes, I will add this later. Thanks.}}The
parameters of interest are the number of samples and the running time
required by the algorithm.

\subsection{Basic Tools} \label{ssec:tools}

We recall \new{some} useful tools from probability theory.

\smallskip

\ignore{Our first tool says that $O(1/\eps^2)$ samples suffice to learn any distribution
within error $\eps$ with respect to the \emph{Kolmogorov distance}.
This fundamental fact is known as the \emph{ Dvoretzky-Kiefer-Wolfowitz
(DKW) inequality} (\cite{DKW56}).}
\ignore{The DKW inequality states that for $m=\Omega((1/\eps^2)\cdot \ln(1/\delta))$,
with probability $1-\delta$ the {empirical distribution} $\wh{p}_m$ will be
$\eps$-close to $p$ in Kolmogorov distance.
This sample bound is asymptotically optimal and independent of the support size.

\begin{theorem}[\cite{DKW56,Massart90}] \label{thm:DKW}
For all $\eps>0$, it holds: $\Pr [ \dk(p, \wh{p}_m) > \eps ] \leq 2e^{-2m\eps^2}.$
\end{theorem} 

}

\new{

\medskip

\noindent {\bf The VC inequality.}
Given $m$ independent samples $s_1,\dots,s_m$, drawn from $p:[n] \to [0,1],$
the {\em empirical distribution} $\wh{p}_m : [n] \to [0,1]$ is defined as
follows: for all $i \in [n]$, $\wh{p}_m(i) = |\{j \in [m] \mid s_j=i\}| / m$.
Fix a family of subsets $\mathcal{A}$ over $[n]$ of VC--dimension $d$. 
The {\em VC inequality} states that for $m=\Omega(d/\eps^2)$,
with probability $9/10$ the {empirical distribution} $\wh{p}_m$ will be
$\eps$-close to $p$ in $\mathcal{A}$-distance.
This sample bound is asymptotically optimal.

\begin{theorem}[VC inequality, {\cite[p.31]{DL:01}}]
\label{thm:vc-inequality}
Let $\widehat{p}_m$ be an empirical distribution of $m$ samples from $p$.
Let $\mathcal A$ be a family of subsets of VC--dimension $d$.
Then
$$ \E \left [ \norm{p - \widehat{p}_m}_{\mathcal A} \right] \leq O(\sqrt{d/m}) .$$
\end{theorem}

\noindent {\bf Uniform convergence.}
We will also use the following uniform convergence bound:
\begin{theorem}[{\cite[p17]{DL:01}}]
\label{thm:bdd-diff}
Let $\mathcal A$ be a family of subsets over $[n]$, and $\widehat{p}_m$ be an
empirical distribution of $m$ samples from $p$.
Let $X$ be the random variable $\norm{p - \widehat{p}_{m}}_{\mathcal A}$.
Then we have
$$ \Pr\left[ X - \E[X] > \eta \right] \leq e^{-2m\eta^2}.$$
\end{theorem}

}

Our second tool, due to Birg\'{e}~\cite{Birge:87b},
provides a sample-optimal and computationally efficient algorithm
to learn monotone distributions to $\eps$-accuracy in total variation distance.
Before we state the relevant theorem, we need a definition.
We say that a distribution $p$ is $\delta$-close to being non-increasing (resp. non-decreasing)
if there exists a non-increasing (resp. non-decreasing) distribution $q$ such that
$\dtv(p,q)\leq \delta$. We are now ready to state  Birg\'{e}'s result:

\begin{theorem}[\cite{Birge:87b}, Theorem 1] \label{thm:birge-monotone}
(semi-agnostic learner) \label{thm:agnostic-monotone-learner}
There is an algorithm $\ANI$ with the following performance guarantee: Given
$m$ independent samples from a distribution $p$  over $[n]$ which is $\opt$-close to being non-increasing,
$\ANI$ performs $\tilde{O}( m \cdot \log n + m^{1/3} \cdot (\log n)^{5/3})$
bit operations
and outputs a (succinct description of a) hypothesis distribution $\widetilde{p}$ over $[n]$ that satisfies
$$\E[\dtv(\widetilde{p},p)] \leq 2 \cdot \opt + O\Big(\big(\log n/(m+1)\big)^{1/3}\Big).$$
The aforementioned algorithm partitions the domain $[n]$
in $O(m^{1/3} \cdot (\log n)^{2/3})$ intervals and outputs a hypothesis distribution
that is uniform within each of these intervals.
\end{theorem}

By taking $m = \Omega (\log n / \eps^3)$, one obtains a hypothesis such that
$\E[\dtv(\widetilde{p},p)] \leq 2 \cdot \opt + \eps.$
We stress that Birg\'{e}'s algorithm for learning non-increasing distributions~\cite{Birge:87b}
is in fact ``semi-agnostic,'' in the sense that it also learns distributions that are close
to being non-increasing; this robustness will be crucial for us later (since in our final
algorithm we will use Birg\'{e}'s algorithm on distributions identified
by our tester, that are close to monotone but not necessarily perfectly
monotone). This semi-agnostic property is not explicitly stated in~\cite{Birge:87b}
but it can be shown to follow easily from his results. We show how
the semi-agnostic property follows from Birg\'{e}'s results in
Appendix~\ref{sec:birge-agnostic}.
Let $\AND$ denote the corresponding semi-agnostic algorithm for learning non-decreasing distributions.

\smallskip

Our final tool is a routine to do \emph{hypothesis testing}, i.e., to select a
high-accuracy hypothesis distribution from a collection of hypothesis distributions one
of which has high accuracy.  The need for such a routine arises in several places; in
some cases we know that a distribution is monotone, but do not know whether it is
non-increasing or non-decreasing. In this case, we can run both algorithms $\AND$ and
$\ANI$ and then choose a good hypothesis using hypothesis testing.  Another need for
hypothesis testing is to ``boost confidence'' that a learning algorithm generates a
high-accuracy hypothesis.  Our initial version of the algorithm for
Theorem~\ref{thm:main} generates an $\eps$-accurate hypothesis with probability at least
$9/10$; by running it $O(\log(1/\delta))$ times using a hypothesis testing
routine,
it is possible to identify an $O(\eps)$-accurate hypothesis with probability $1-\delta.$
Routines of the sort that we require have been given in e.g., \cite{DL:01} and
\cite{DDS12stoclong}; we use the following theorem from~\cite{DDS12stoclong}:


\begin{theorem} \label{thm:choosehypothesis}
There is an algorithm {\tt Choose-Hypothesis}$^p({h}_1,{h}_2,\eps',\delta')$ which is given
\new{sample} access to $p$, two hypothesis distributions $h_1,h_2$ for $p$,  an accuracy parameter
$\eps'$, and a confidence parameter $\delta'.$  It makes $m=O(\log(1/\delta')/\eps'^2)$ draws
from $p$ and returns a hypothesis $h \in \{h_1,h_2\}.$  If one of $h_1,h_2$ has
$\dtv(h_i,p) \leq \eps'$ then with probability $1-\delta'$ the hypothesis $h$ that
{\tt Choose-Hypothesis} returns has $\dtv(h,p) \leq 6 \eps'.$
\end{theorem}


\noindent For the sake of completeness, we describe and analyze the {\tt Choose-Hypothesis} algorithm
in Appendix~\ref{sec:choosehypothesis}.



\section{Learning $k$-modal Distributions} \label{sec:learning}


In this section, we present our main result: a nearly sample-optimal
and computationally efficient algorithm to learn an unknown $k$-modal
distribution.
In Section~\ref{ssec:kmodal-learn-simple} we present a simple learning
algorithm with a suboptimal sample complexity.
In Section~\ref{ssec:kmodal-learn} we present our main result which involves
a property testing algorithm as a subroutine.


\ignore{

I really feel we should sell our algorithms as reductions to the monotone case.
What is the idea of both algortithms ? Isolate the modes. You have at most $k$ of them.
If you can isolate them in $k$ interval whose total mass under $p$ is at most $\eps$,
you are good.You can use Birge to learn the remaining stuff and ignore these bad intervals.

The first algorithm: Construct $k/\eps$ intervals of mass at most $\eps/k$ each.
You know $k$ of them are bad, but you do not know which ones. The bad intervals have
total mass $\eps$. Use Birge to learn each interval. $k$ of the hypotheses will be wacky,
the rest will be good.Does not matter you don't know which ones.

The second algorithm: In fact, identify the $k$ bad intervals. Testing algotihm:
Test whether an interval is monotone or not. Partition distribution in $k/\eps$ intervals of mass $\eps/k$.
Test (using binary search) to find the $k$ non-monotone intervals. Ignore those.
Learn the rest using Birge. Why do we win? We run Birge in $k$ superintervals,
as opposed to $k/\eps$ atomic intervals.

Issue: Testing algorithm takes samples from $p$. This means we do need some kind
of coupon collector to deduce that all these bins get enough samples. Probably will
add a $\log k$ or sth. But I feel it should be ok.
}

\subsection{Warm-up: A simple learning algorithm} \label{ssec:kmodal-learn-simple}

In this subsection, we give an algorithm that runs in time $\poly(k, \log n, 1/\eps, \log(1/\delta))$
and learns an unknown $k$-modal distribution to accuracy $\eps$ and confidence $\delta$.
The sample complexity of the algorithm is \new{essentially optimal as a function of $k$ (up to a logarithmic factor), but}
suboptimal as a function of $\eps$, by a polynomial factor.

In the following \new{pseudocode} we give a detailed description of the algorithm {\tt Learn-kmodal-simple};
\new{the algorithm outputs} an $\eps$-accurate hypothesis with confidence $9/10$ (see Theorem~\ref{thm:correctness-learn-simple}). 
We explain how to boost the confidence to $1-\delta$ after the proof of the theorem.

The algorithm  {\tt Learn-kmodal-simple} works as follows: We start by partitioning
the domain $[n]$ into consecutive intervals of mass ``approximately $\eps/k$.''
To do this, \new{we draw $\Theta(k/\eps^3)$ samples from $p$ and greedily partition the domain \newer{into} disjoint intervals of 
empirical mass roughly $\eps/k$.} (Some care is needed in this step, since there may be ``heavy'' points in the support of the distribution;
however, we gloss over this technical issue for the sake of this intuitive explanation.)
\new{Note that we do {\em not} have a guarantee that {\em each} such interval will have true probability mass $\Theta(\eps/k)$. 
In fact, it may well be the case that the additive error $\delta$ between the true probability mass of an interval and its empirical mass
(roughly $\eps/k$) is $\delta = \omega(\eps/k)$.
The error guarantee of the partitioning is more ``global''  in that the {\em sum} of these errors across all such intervals 
is at most $\eps$. In particular,  as a simple corollary of the VC inequality, we can deduce 
the following statement that will be used several times throughout the paper:

\begin{fact} \label{fact:vc}
Let $p$ be any distribution over $[n]$ and $\widehat{p}_m$ be the empirical distribution of $m$ samples from $p$.  
For $m = \Omega\left( (d/\eps^2) \log(1/\delta) \right)$, with probability at least $1-\delta$, 
for {\em any} collection $\mathcal{J}$ of (at most) $d$ disjoint intervals in $[n]$, we have that 
$$ \littlesum_{J \in \mathcal{J}} |p(J) - \widehat{p}_m(J)|  \le \eps.$$
\end{fact}

\begin{proof}
Note that 
\begin{equation} \label{eq:ena}
\littlesum_{J \in \mathcal{J}} |p(J) - \widehat{p}_m(J)|   = 2 |p(A) - \widehat{p}_m(A) |,
\end{equation}
where $A = \{J \in \mathcal{J} : p(J) >  \widehat{p}_m(J) \}$. Since $\mathcal{J}$ is a collection of at most $d$ intervals, 
it is clear that $A$ is a union of at most $d$ intervals.
If $\mathcal{A}_d$ is the family of all unions of at most $d$ intervals, then the right hand side of (\ref{eq:ena}) is at most $2 \| p - \widehat{p}_m \|_{\mathcal{A}_d}$.
Since the VC--dimension of $\mathcal{A}_d$ is $2d$, Theorem~\ref{thm:vc-inequality} implies that the quantity (\ref{eq:ena}) has expected value at most \newer{$\eps/2$}.
The claim now follows by an application of Theorem~\ref{thm:bdd-diff} with \newer{$\eta = \eps/2$}.
\end{proof}

}

If this step is successful, we have partitioned the domain into a set of $O(k/\eps)$ consecutive intervals
of probability mass ``roughly $\eps/k$.'' \new{The} next step is to apply Birg{\'e}'s monotone learning algorithm
to each interval.

\smallskip

A caveat comes from the fact that not all such intervals are guaranteed to be
monotone (or even close to being monotone). However, since our input distribution is assumed to be $k$-modal,
all but (at most) $k$ of these intervals are monotone. Call a non-monotone interval ``bad.'' Since all intervals
have \new{empirical} probability mass at most $\eps / k$ and there are at most $k$ bad intervals, it follows from Fact~\ref{fact:vc} that 
these intervals contribute at most \new{$O(\eps)$} to the total mass. So
even though Birg{\'e}'s algorithm gives no guarantees for bad intervals,
these intervals do not affect the error by more than \new{$O(\eps)$}.

Let us now focus on the  monotone intervals. For each such interval, we do not
know if it is monotone increasing or monotone decreasing. To overcome this
difficulty, we run both monotone algorithms $\ANI$ and $\AND$
for each interval and then use
hypothesis testing to choose the correct candidate distribution.

\smallskip

Also, note that since we have $\new{O(k/\eps)}$ intervals, we need to run each instance of both
the monotone learning algorithms and the hypothesis testing algorithm with
confidence $1-O(\eps/k)$, so that we can guarantee that the
overall algorithm has confidence $9/10$. Note that
Theorem~\ref{thm:birge-monotone} and Markov's inequality imply that if we
draw $\Omega(\log n / \eps^3)$ samples from a non-increasing distribution $p$,
the hypothesis $\widetilde{p}$ output by $\ANI$ satisfies
$\dtv(\widetilde{p},p) \leq \eps$ with probability $9/10$. We can boost the
confidence to $1-\delta$ with an overhead of $O(\log(1/\delta)
\log\log(1/\delta))$ in the sample complexity:


\begin{fact} \label{fact:birge-tournament}
Let $p$ be a non-increasing distribution over $[n]$.
There is an algorithm $\ANI_{\delta}$ with the following performance guarantee:
Given $(\log n /\eps^3) \cdot \tilde{O}(\log(1/\delta)) )$ samples from $p$, $\ANI_{\delta}$ performs
$\tilde{O}\left( (\log^2 n / \eps^3)\cdot \log^\new{2}(1/\delta) \right)$ bit operations
and outputs a (succinct description of a) hypothesis distribution $\widetilde{p}$ over $[n]$
that satisfies $\dtv(\widetilde{p}, p) \leq \eps$ with probability at least $1-\delta$.
\end{fact}


The algorithm $\ANI_{\delta}$ runs $\ANI$ $O(\log(1/\delta))$ times
and performs a tournament among the candidate hypotheses using
{\tt Choose-Hypothesis}.
Let $\AND_{\delta}$ denote the corresponding  algorithm for learning
non-decreasing distributions with confidence $\delta$.
We postpone further details on these algorithms to  Appendix~\ref{sec:ht}.

\begin{theorem} \label{thm:correctness-learn-simple}
The algorithm {\small {\tt Learn-kmodal-simple}}
uses  \new{$$\frac{k \log n}{\eps^4} \cdot \tilde{O}\left( \log (k/\eps) \right)$$} samples,
performs
$\poly(k, \log n, 1/\eps)$ bit operations,
and learns a $k$-modal distribution to accuracy $O(\eps)$ with probability $9/10$.
\end{theorem}

\bigskip

\hskip-.2in \framebox{
\medskip \noindent \begin{minipage}{16.5cm}

\medskip

{\tt Learn-kmodal-simple}

\noindent {\bf Inputs:} $\eps >0$; sample access to $k$-modal
distribution $p$ over $[n]$


\begin{enumerate}

\item \new{Fix $d: = \lceil 20k/\eps \rceil$. Draw $ r = \Theta(d/\eps^2)$} samples
        from $p$ and let $\wh{p}$ denote the resulting empirical distribution.


\item Greedily partition the domain $[n]$ into $\ell$ {\em atomic intervals} $\mathcal{I}:= \{I_i\}_{i=1}^{\ell}$ as follows:
          \begin{enumerate}
          \item[(a)] $I_1 := [1, j_1]$, where $j_1 := \min\{ j \in [n] \mid \wh{p}([1,j]) \geq \eps/(10k)\}$.
           \item[(b)] For $i \geq 1$, if $\cup_{j=1}^i I_j = [1, j_i]$, then $I_{i+1}:=[j_i+1, j_{i+1}]$, where
           $j_{i+1}$ is defined as follows: 
           \begin{itemize}
          \item If $\wh{p}([j_i+1, n]) \geq \eps/(10k)$, then $j_{i+1}: = \min\{ j \in [n] \mid \wh{p}([j_{i}+1, j]) \geq \eps/(10k)\}${.}
           \item Otherwise, $j_{i+1} := n$.
           \end{itemize}
           \end{enumerate}

\item Construct a set of $\ell$ {\em light intervals}  $\mathcal{I'}:= \{I'_i\}_{i=1}^{\ell}$ and a set $\{b_i\}_{i=1}^t$ of $t \leq \ell$ {\em heavy points} as follows:
          \begin{enumerate}
             \item[(a)] For each interval $I_i = [a,b] \in \mathcal{I}$, if $\wh{p}(I_{\new{i}}) \geq \eps/(5k)$ define $I'_i:= [a,b-1]$ and make $b$ a heavy point.
                             (Note that it is possible to have $I'_i=\emptyset.$)
             \item[(b)] Otherwise, define $I'_i :=I_i$.
          \end{enumerate}

\smallskip

\noindent Fix $\delta': = {\eps / (500k)}$.


\item Draw $m =  (k/\eps^4) \cdot \log(n) \cdot \tilde{\Theta}(\log(1/\delta'))$
samples $\mathbf{s} = \{s_i\}_{i=1}^m$ from $p$. For each light interval $I'_i$, $i \in [\ell]$, run both
$\ANI_{\delta'}$ and $\AND_{\delta'}$ on the conditional distribution $p_{I'_i}$ using the samples in $\mathbf{s} \cap I'_i$. 
Let   $\widetilde{p}^{\downarrow}_{I'_i}$, $\widetilde{p}^{\uparrow}_{I'_i}$  be the corresponding conditional hypothesis distributions.


\item Draw $m' = \Theta((k/\eps^4)\cdot \log(1/\delta'))$ samples  $\mathbf{s'} = \{s'_i\}_{i=1}^{m'}$ from $p$.
For each light interval $I'_i$, $i \in [\ell]$,
run  {\tt Choose-Hypothesis}$^p(\widetilde{p}^{\uparrow}_{I'_i},\widetilde{p}^{\downarrow}_{I'_i}, \eps,\delta')$
using the samples in $\mathbf{s'} \cap I'_i$.
Denote by $\widetilde{p}_{I'_i}$ the returned conditional distribution on $I'_i$.

\item Output the hypothesis $h = \littlesum_{j=1}^{\ell} \wh{p}(I'_j) \cdot \widetilde{p}_{I'_j} + \littlesum_{j=1}^t \wh{p}(b_j) \cdot {\bm 1}_{b_j}$.

\end{enumerate}

\end{minipage}}

\bigskip

\begin{proof} 
\new{First, it is easy to see that the algorithm has the claimed sample complexity.
Indeed, the algorithm draws a total of $r+m+m'$ samples in Steps 1, 4 and 5. 
The running time is also easy to analyze, as it is easy to see that
every step can be performed in polynomial time (in fact, nearly linear time)
in the sample size.}

We need to show that with probability $9/10$
(over its random samples), algorithm
{\tt Learn-kmodal-simple} outputs a hypothesis $h$ such that $\dtv(h,p) \leq O(\eps)$.

\new{Since $r = \Theta(d/\eps^2)$ samples are drawn in Step~1, Fact~\ref{fact:vc}
implies that with probability of failure at most $1/100$,
for each family $\mathcal{J}$ of at most $d$ disjoint intervals from $[n]$, we have
\begin{equation} \label{eqn:cond}
\littlesum_{J \in \mathcal{J}} |p(J) - \widehat{p}_m(J)|  \le \eps.
\end{equation}
}
For the rest of the analysis of {\tt Learn-kmodal-simple}
we condition on this ``good'' event.

\smallskip

Since every atomic interval $I \in \mathcal{I}$ has $\wh{p}(I) \geq
\eps/(10k)$ (except potentially the rightmost one), it follows that the
number $\ell$ of atomic intervals constructed in Step~2 satisfies $\ell \leq
10 \cdot (k/\eps)$.  \new{By the construction in Steps 2 and 3,
every light interval $I'  \in \mathcal{I}'$ has $\wh{p}(I')  \leq \eps/(5k)$.}
Note also that every heavy point $b$ has $\wh{p}(b) \geq \eps/(10k)$ and
the number of heavy points $t$ is at most $\ell.$


Since the light intervals and heavy points form a partition of $[n]$, we
can write $$p =  \littlesum_{j=1}^{\ell} p(I'_j) \cdot p_{I'_j} + \littlesum_{j=1}^t p(b_j) \cdot {\bm 1}_{b_j}.$$
Therefore, we can bound the variation distance as follows:
\begin{equation} \label{eqn:dtv-ub}
\dtv(h,p) \leq  \littlesum_{j=1}^{\ell} |  \wh{p}(I'_j)  - p(I'_j) | +  \littlesum_{j=1}^t |  \wh{p}(b_j) - p(b_j)| 
 + \littlesum_{j=1}^{\ell} p(I'_j) \cdot \dtv (\widetilde{p}_{I'_j}, p_{I'_j}).
\end{equation}
\new{Since $\ell+t \le d$, by Fact~\ref{fact:vc} and our conditioning, the contribution of the first two terms 
to the sum is upper bounded by $\eps$.}

We proceed to bound the contribution of the third term. Since $p$ is $k$-modal,
at most $k$ of the light intervals $I'_j$ are not monotone for $p.$ Call
these intervals ``bad'' and denote by $\mathcal{B}$ as the set of bad intervals.
Even though we have not identified the bad intervals, we know that all such
intervals are light. Therefore, \new{their total empirical probability mass (under $\widehat{p}_m$)
is at most $k \cdot \eps/(5k) = \eps/5$, i.e., $\littlesum_{I \in \mathcal{B}} \widehat{p}(I) \le \eps/5$. 
By our conditioning \newer{(see Equation~(\ref{eqn:cond}))} and the triangle inequality it follows that 
$$\left|  \littlesum_{I \in \mathcal{B}} p(I)  -  \littlesum_{I \in \mathcal{B}} \widehat{p}(I)  \right|  \le   \littlesum_{I \in \mathcal{B}} \left| p(I) - \widehat{p}(I)  \right| \le \eps$$
which implies that the true probability mass of the bad intervals is at most $\eps/5+\eps = 6\eps/5$.
Hence, the contribution of bad intervals to the third term of the right hand side of (\ref{eqn:dtv-ub}) is at
most $O(\eps)$.}  (Note that this statement holds true independent of the
samples $\mathbf{s}$ we draw in Step 4.) 

It remains to bound the contribution of monotone intervals to the third term.
Let $\ell' \leq \ell$ be the number of monotone light intervals and
assume after renaming the indices that they are $\widetilde{\mathcal{I}} : = \{I'_j\}_{j=1}^{\ell'}$.
To bound from above the right hand side of (\ref{eqn:dtv-ub}), it suffices to show that with probability
at least $19/20$ (over the samples drawn in Steps 4-5) it holds
\begin{equation} \label{eqn:third-term}
\littlesum_{j=1}^{\ell'} p(I'_j) \cdot \dtv (\widetilde{p}_{I'_j}, p_{I'_j}) = O(\eps).
\end{equation}

\new{\noindent To prove (\ref{eqn:third-term}) we
partition the set $\widetilde{\mathcal{I}}$ \newer{into} three subsets based on their probability mass under $p$. 
Note that we do not have a lower bound
on the probability mass of intervals in $\widetilde{\mathcal{I}}$.
Moreover, by our conditioning \newer{(see Equation~(\ref{eqn:cond}))} and \newer{the fact that each interval in $\widetilde{\mathcal{I}}$ is light,} it follows that any $I \in \widetilde{\mathcal{I}}$ has 
$p(I) \le \widehat{p}(I)+\eps \le 2\eps.$ We define the partition of $\widetilde{\mathcal{I}}$ into the following three sets:
$\widetilde{\mathcal{I}_1} = \{ I \in  \widetilde{\mathcal{I}}: p(I) \le \eps^2/(20k)\},$
$\widetilde{\mathcal{I}_2} = \{ I \in  \widetilde{\mathcal{I}}: \eps^2/(20k) < p(I) \le \eps/k\}$ and
$\widetilde{\mathcal{I}_3} = \{ I \in  \widetilde{\mathcal{I}}:  \eps/k < p(I) \le  2\eps \}.$ 

We bound the contribution of each subset in turn.
It is clear that the contribution of  $\widetilde{\mathcal{I}_1}$
to (\ref{eqn:third-term}) is at most 

$$\littlesum_{I \in \widetilde{\mathcal{I}_1}} p(I) \le |\widetilde{\mathcal{I}_1}| \cdot \eps^2/(20k) \le   \ell' \cdot \eps^2/(20k) \leq \ell \cdot \eps^2/(20k) \leq \eps/2.$$

}

To bound from above the contribution of \new{$\widetilde{\mathcal{I}_2}$} to (\ref{eqn:third-term}),
we partition \new{$\widetilde{\mathcal{I}_2}$} into \new{$g_2 = \lceil \log_2 (20/\eps) \rceil =  \Theta (\log(1/\eps))$} groups.
For $i \in [\new{g_2}]$, the set $(\widetilde{\mathcal{I}_2})^i$
consists of those intervals in  $\widetilde{\mathcal{I}_2}$
that have mass under $p$ in the range
\new{$\left(  2^{-i} \cdot (\eps/k), 2^{-i+1} \cdot (\eps/k) \right]$}.
\new{The following statement establishes the variation distance closeness between the conditional hypothesis
for an interval in the $i$-th group $(\widetilde{\mathcal{I}_2})^i$ and the corresponding conditional distribution.}
\begin{claim} \label{claim:not-too-light}
With probability at least $19/20$ (over the sample $\mathbf{s}, \mathbf{s'}$),
for each $i \in [\new{g_2}]$ and each monotone light interval $I'_j \in \new{(\widetilde{\mathcal{I}_2})^i}$ we have
$d_{TV}(\widetilde{p_{I'_j}},p_{I'_j}) = O(\new{2^{i/3} \cdot \eps}).$
\end{claim}
\begin{proof} 
Since in Step~4 we draw $m$ samples, and each interval $I'_j \in
\new{(\widetilde{\mathcal{I}_2})^i}$ has $p(I'_j) \in  \new{ \left[ 2^{-i} \cdot (\eps/k), 2^{-i+1} \cdot (\eps/k) \right]}$,
a standard coupon collector argument~\cite{NS:60} tells us that with probability $99/100$, for {\em each}
$(i,j)$ pair, the interval $I'_j$ will get at least
$\new{2^{-i} \cdot (\log (n) / \eps^3)} \cdot \tilde{\Omega}( \log(1/\delta'))$ many samples.
Let's rewrite this as \new{$(\log (n) / (2^{i/3} \cdot \eps)^3) \cdot \tilde{\Omega}( \log(1/\delta'))$} samples.
We condition on this event.

Fix an interval $I'_j \in \new{(\widetilde{\mathcal{I}_2})^i}$. We first show that with failure probability
at most  $\eps/(500k)$ after Step~4, either $\widetilde{p}^{\downarrow}_{I'_j}$ or
$\widetilde{p}^{\uparrow}_{I'_j}$ will be \new{$(2^{i/3} \cdot \eps)$}-accurate.
Indeed, by Fact~\ref{fact:birge-tournament} and taking into account the number of samples
that landed in $I'_j$, with probability $1-\eps/(500k)$ over $\mathbf{s}$,
$ \dtv (\widetilde{p}^{\alpha_i}_{I'_j}, p_{I'_j})  \leq \new{2^{i/3} \eps}$,
where $\alpha_i = \downarrow$ if $p_{I'_j}$ is non-increasing and  $\alpha_i = \uparrow$ otherwise.
By a union bound over all (at most $\ell$ many) $(i,j)$ pairs,
it follows that with probability at least $49/50$, for each interval $I'_j \in \new{(\widetilde{\mathcal{I}_2})^i}$
one of the two candidate hypothesis distributions is $\new{(2^{i/3} \eps)}$-accurate. We condition on this event.

\newer{
Now consider Step~5.  Since this step draws $m'$ samples, and each interval $I'_j \in (\widetilde{{\cal I}}_2)^i$ has $p(I'_j) \in \left( 2^{-i} \cdot (\eps/k), 2^{-i+1} \cdot (\eps/k) \right]$, as before a standard coupon collector argument~\cite{NS:60} tells us that with probability $99/100$, for \emph{each} $(i,j)$ pair, the interval $I'_j$ will get at least 
$(1/(2^{i/3} \cdot \eps)^3) \cdot \tilde{\Omega}(\log(1/\delta'))$ many samples in this step; we henceforth assume that this is indeed the case for each $I'_j$.  
Thus, Theorem~\ref{thm:choosehypothesis} applied to each fixed interval $I'_j$ implies that the algorithm {\tt Choose-Hypothesis}  }
will output a hypothesis that is $\new{6\cdot(2^{i/3} \eps)}$-close to $p_{I'_j}$ with probability $1-\eps/(500k)$.
By a union bound, it follows that with probability at least $49/50$,
the above condition holds for all monotone light intervals under consideration.
Therefore, except with failure probability $19/20$, the statement of the 
{claim} holds. 
\end{proof}

Given the claim, we exploit the fact that for intervals $I'_j$ such that $p(I'_j)$ is small
we can afford larger error on the total variation distance.
More precisely, let \new{$c_i  = |(\widetilde{\mathcal{I}_2})^i|$}, the number of intervals in
\new{$(\widetilde{\mathcal{I}_2})^i$}, and note that \new{$\littlesum_{i=1}^{g_2} c_i \leq \ell$}.
Hence, we can bound the contribution of \new{$\widetilde{\mathcal{I}_2}$} to (\ref{eqn:third-term}) by
\new{
\begin{eqnarray*}
\littlesum_{i=1}^{g_2} c_i  \cdot (\eps/k) \cdot 2^{-i+1} \cdot  O(2^{i/3} \cdot \eps)
\leq O(1) \cdot (2\eps^2/k) \cdot \littlesum_{i=1}^{g_2} c_i \cdot 2^{-2i/3}.
\end{eqnarray*}
}
Since  \new{$\littlesum_{i=1}^{g_2} c_i = |\widetilde{\mathcal{I}_2}|   \leq \ell$}, the above expression is maximized for
\new{$c_1= |\widetilde{\mathcal{I}_2}| \le \ell$ and $c_i = 0$, $i >1$,} and the maximum value is at most
\new{$$O(1) \cdot (\eps^2/k) \cdot \ell = O(\eps).$$} 
\new{
Bounding the contribution of $\widetilde{\mathcal{I}_3}$ to (\ref{eqn:third-term}) is very similar.
We partition $\widetilde{\mathcal{I}_3}$ into $g_3 = \lceil \log_2 k \rceil +1 =  \Theta (\log(k))$ groups.
For $i \in [\new{g_3}]$, the set $(\widetilde{\mathcal{I}_3})^i$
consists of those intervals in  $\widetilde{\mathcal{I}_3}$
that have mass under $p$ in the range $\left( 2^{-i+1} \cdot \eps, 2^{-i+2} \cdot \eps \right]$.
The following statement is identical to Claim~\ref{claim:not-too-light} 
albeit with different parameters:
\begin{claim} \label{claim:heavy}
With probability at least $19/20$ (over the sample $\mathbf{s}, \mathbf{s'}$),
for each $i \in [g_3]$ and each monotone light interval $I'_j \in (\widetilde{\mathcal{I}_3})^i${,} we have
$d_{TV}(\widetilde{p_{I'_j}},p_{I'_j}) = O(2^{i/3} \cdot \eps \cdot k^{-1/3}).$
\end{claim}
Let \new{$f_i  = |(\widetilde{\mathcal{I}_3})^i|$}, the number of intervals in $(\widetilde{\mathcal{I}_3})^i$. 
Each interval $I \in (\widetilde{\mathcal{I}_3})^i$ has $p(I) \in (d_i, 2d_i]$, where $d_i : = 2^{-i+1} \cdot \eps$.
We therefore have
\begin{equation} \label{eqn:total-mass-at-most-one}
\littlesum_{i=1}^{g_3} d_i \newer{f_i} \le p(\widetilde{\mathcal{I}_3}) \le 1.
\end{equation}
We can now bound from above the contribution of $\widetilde{\mathcal{I}_3}$ to (\ref{eqn:third-term}) by
\begin{eqnarray*}
\littlesum_{i=1}^{g_3} 2d_i \newer{f_i} \cdot  O(2^{i/3} \cdot \eps \cdot k^{-1/3})
\leq O(1) \cdot (\eps/k^{1/3}) \cdot \littlesum_{i=1}^{g_3} d_i \newer{f_i} \cdot 2^{i/3}.
\end{eqnarray*}
By (\ref{eqn:total-mass-at-most-one}) it follows that the above expression is maximized for
$d_{g_3} \newer{f_{g_3}} = 1$ and $d_i \newer{f_i}= 0$, $i <g_3$. The maximum value is at most
$$O(1) \cdot (\eps/k^{1/3}) \cdot 2^{g_3/3} = O(\eps)$$
where the final equality uses the fact that $2^{g_3} \le 4k$ as follows by our definition of $g_3$. 
}
This proves  (\ref{eqn:third-term}) and completes the proof of Theorem~\ref{thm:correctness-learn-simple}.
\end{proof}

To get an $O(\eps)$-accurate hypothesis with probability $1-\delta$, we
can simply run {\tt Learn-kmodal-simple} $O(\log(1/\delta))$ times
and then perform a tournament using Theorem~\ref{thm:choosehypothesis}.
This increases the sample complexity by a $\tilde{O}(\log(1/\delta))$ factor.
The running time increases by a factor of $O(\log^2(1/\delta)).$
We postpone the details for Appendix~\ref{sec:ht}.

\subsection{Main Result: Learning $k$-modal distributions using testing} \label{ssec:kmodal-learn}

Here is some intuition to motivate our $k$-modal distribution learning algorithm
and give a high-level idea of why the dominant term in its sample complexity is
$O(k\log(n/k)/\eps^3).$

Let $p$ denote the target $k$-modal distribution to be learned.
As discussed above, optimal (in terms of time and sample complexity) algorithms are known
for learning a monotone distribution over $[n]$, so if the locations of the
$k$ modes of $p$ were known then it would be straightforward to learn $p$ very efficiently
by running the monotone distribution learner over $k+1$ separate intervals.  But it is clear that
in general we cannot hope to efficiently identify the modes of $p$ exactly (for instance it could
be the case that $p(a)=p(a+2)=1/n$ while $p(a+1) = 1/n + 1/2^n$).  Still, it is natural to try to
decompose the $k$-modal distribution into a collection of (nearly) monotone  distributions
and learn those.  At a high level that is what our algorithm does, using a novel \emph{property testing}
algorithm.


More precisely, we give a distribution testing algorithm with the following performance guarantee:
Let $q$ be a $k$-modal distribution over $[n]$.  Given an accuracy parameter $\tau$, our tester takes $\poly(k/\tau)$ samples from $q$ and outputs ``yes''
with high probability if $q$ is {monotone} and ``no'' with high probability
if $q$ is $\tau$-far from every monotone distribution.
(We stress that the assumption that $q$ is $k$-modal is essential here, since
an easy argument given in \cite{BKR:04long} shows that $\Omega(n^{1/2})$ samples are required to test
whether a general distribution over $[n]$ is monotone versus $\Theta(1)$-far from monotone.)

With some care, by running the above-described tester $O(k/\eps)$ times with accuracy parameter $\tau$,
we can decompose the domain $[n]$ into


\begin{itemize}

\item at most $k+1$ ``superintervals,'' which have the property that the conditional distribution of
$p$ over each superinterval is almost monotone ($\tau$-close to monotone)\ignore{and each superinterval
has probability mass at least $\Omega(\eps/k)$ under $p$};


\item at most $k+1$ ``negligible intervals'', which have the property that each one has
probability mass at most $O(\eps/k)$ under $p$ (so ignoring all of them incurs at most $O(\eps)$ total
error); and


\item at most $k+1$ ``heavy'' points, each of which has mass at least $\Omega(\eps/k)$ under $p.$

\end{itemize}


\noindent
We can ignore the negligible intervals, and the heavy points are easy to handle; however some care must be taken to learn the ``almost monotone''
restrictions of $p$ over each superinterval.  A naive approach, using a generic
$\log(n)/\eps^3$-sample monotone distribution learner that has no performance guarantees if the target
distribution is not monotone, leads to an inefficient overall algorithm.  Such an approach
would require that $\tau$ (the closeness parameter used by the tester) be at most $1/$(the sample complexity
of the monotone distribution learner), i.e., $\tau < \eps^3/\log(n)$.  Since the sample
complexity of the tester is $\poly(k/\tau)$ and the tester is run $\Omega(k/\eps)$ times, this approach would lead
to an overall sample complexity that is unacceptably high.

Fortunately, instead of using a generic monotone distribution learner, we
can use the semi-agnostic monotone distribution learner of Birg\'{e}
(Theorem~\ref{thm:birge-monotone}) that can handle
deviations from monotonicity far more efficiently than the above naive
approach.  Recall that given draws from a distribution
$q$ over $[n]$ that is $\tau$-close to monotone, this algorithm uses
$O(\log(n)/\eps^3)$ samples and outputs a hypothesis distribution that is
$(2\tau + \eps)$-close to monotone.  By using
this algorithm we can take the accuracy parameter $\tau$ for our tester to be $\Theta(\eps)$ and learn
the conditional distribution of $p$ over a given superinterval
to accuracy $O(\eps)$ using $O(\log(n)/\eps^3)$ samples from that
superinterval. Since there are $k+1$ superintervals overall, a careful analysis shows that
$O(k \log(n)/\eps^3)$ samples suffice to handle all the superintervals.

We note that the algorithm also requires an additional additive
$\poly(k/\eps)$ samples (independent of $n$) besides this dominant term
(for example, to run the tester and to estimate accurate weights with
which to combine the various sub-hypotheses).  The overall sample complexity we achieve is stated in Theorem~\ref{thm:learnkmodal} below.


\begin{theorem} (Main) \label{thm:learnkmodal}
The algorithm {\small {\tt Learn-kmodal}} uses
$$O\left(k\log(n/k)/\eps^3 + (\new{k^2}/\eps^3)\cdot \log(k/\eps) \cdot \log\log(k/\eps) \right)$$
samples, performs  $\poly(k, \log n, 1/\eps)$
bit operations, and learns any $k$-modal distribution to accuracy $\eps$
and confidence $9/10$.
\end{theorem}

\ignore{
Birg\'{e} \cite{Birge:87} has shown that any algorithm for learning
a non-increasing distribution over $[n]$ to total
variation distance $\eps$ must use $\Omega(\log(n)/\eps^3)$ samples.
This lower bound can easily be extended to show that any algorithm
for learning a $k$-modal distribution over $[n]$ to accuracy $\eps$
must use $\Omega(\log(n/k)/\eps^3)$ samples.
Consequently, for $k \leq \widetilde{O}(\sqrt{\log n})$
the sample complexity of Theorem~\ref{thm:learnkmodal}
is information-theoretically optimal up to constant factors.
\rnote{Not quite true because of the extra $\log(1/\eps)$, right?}
}

Theorem~\ref{thm:main} follows from Theorem~\ref{thm:learnkmodal}
by running {\tt Learn-kmodal} $O(\log(1/\delta))$ times and
using hypothesis testing to boost the confidence to $1-\delta$.
 We give details in Appendix~\ref{sec:ht}.

Algorithm {\tt Learn-kmodal} makes essential use of
an algorithm $\TND$ for testing whether a $k$-modal distribution over $[n]$
is non-decreasing.
Algorithm $\TND(\eps,\delta)$ uses $O(\log(1/\delta)) \cdot \new{(k/\eps^2)}$ samples from a $k$-modal
distribution $p$ over $[n]$, and behaves as follows:

\begin{itemize}
\item (Completeness) If $p$ is non-decreasing, then $\TND$ outputs ``yes'' with probability at least $1-\delta$;

\item (Soundness) If $p$ is $\eps$-far from non-decreasing, then $\TND$ outputs ``yes'' with probability at most $\delta$.
\end{itemize}

\noindent Let $\TNI$ denote the analogous algorithm for testing whether a $k$-modal distribution
over $[n]$ is non-increasing (we will need both algorithms). The description and proof of correctness for $\TND$ is postponed to the following subsection (Section~\ref{ssec:test}).

\subsection{Algorithm {\tt Learn-kmodal} and its analysis} \label{ssec:learn-main}

Algorithm {\tt Learn-kmodal} is given below with its analysis following.

\smallskip

\hskip-.2in \framebox{
\medskip \noindent \begin{minipage}{16.5cm}

\medskip

{\tt Learn-kmodal}

\smallskip

\noindent {\bf Inputs:} $\eps> 0$; sample access to $k$-modal
distribution $p$ over $[n]$

\smallskip

\begin{enumerate}

\item Fix $\tau:=\eps/(100k)$. Draw $ r = \Theta(1/\tau^2)$ samples from $p$ and let $\wh{p}$ denote the empirical distribution.

\item Greedily partition the domain $[n]$ into $\ell$ {\em atomic intervals} $\mathcal{I}:= \{I_i\}_{i=1}^{\ell}$ as follows:
          \begin{enumerate}
          \item[(a) ]$I_1 := [1, j_1]$, where $j_1 := \min\{ j \in [n] \mid \wh{p}([1,j]) \geq \eps/(10k)\}$.
           \item[(b)] For $i \geq 1$, if $\cup_{j=1}^i I_j = [1, j_i]$, then $I_{i+1}:=[j_i+1, j_{i+1}]$, where
           $j_{i+1}$ is defined as follows: 
           \begin{itemize}
           \item If $\wh{p}([j_i+1, n]) \geq \eps/(10k)$, then $j_{i+1}: = \min\{ j \in [n] \mid \wh{p}([j_{i}+1, j]) \geq \eps/(10k)\}${.}
           \item Otherwise, $j_{i+1} := n$.
           \end{itemize}           
          \end{enumerate}

\item  Set $\tau' := \eps/(2000k)$.
Draw $r' = \Theta((\new{k^2} / \eps^3) \cdot \log(1/\tau') \log\log(1/\tau'))$ samples
$\mathbf{s}$ from $p$ to use in Steps 4-5.

\item  Run both $\TND(\eps,\tau')$ and $\TNI(\eps,\tau')$ over $p_{\cup_{i=1}^{j} I_i}$ for $j=1,2,\ldots$, to find
the leftmost atomic interval $I_{j_1}$ such that both $\TND$ and $\TNI$ return ``no'' over $p_{\cup_{i=1}^{j_1} I_i}$.

Let $I_{j_1} = [a_{j_1},b_{j_1}].$  We consider two cases:

    {\em Case 1:} If $\widehat{p}[a_{j_1},b_{j_1}]
    \geq 2 \eps/(10k)$, define $I'_{j_1}: = [a_{j_1},b_{j_1}-1]$ and
$b_{j_1}$ is a \emph{heavy} point.

   {\em Case 2:}  If $\widehat{p}[a_{j_1},b_{j_1}] < 2 \eps/(10k)$ then define $I'_{j_1}:=I_{j_1}.$

   Call $I'_{j_1}$ a \emph{negligible} interval.
  If $j_1 >1$ then define the first \emph{superinterval} $S_1$ to be $\cup_{i=1}^{j_1-1} I_i$, and
   set $a_1 \in \{\uparrow,\downarrow\}$ to be $a_1=\uparrow$ if $\TND$
   returned ``yes'' on $p_{\cup_{i=1}^{j_1-1} I_i}$ and to be $a_1=\downarrow$ if $\TNI$
   returned ``yes'' on $p_{\cup_{i=1}^{j_1-1} I_i}$.

\item Repeat Step~3 starting with the next interval $I_{j_1+1}$,
i.e., find the leftmost atomic interval $I_{j_2}$ such that
both $\TND$ and $\TNI$ return ``no'' over
$p_{\cup_{i=j_1+1}^{j_2} I_i}.$
  Continue doing this until
   all intervals through $I_{\ell}$ have been used.


   Let $S_1,\dots,S_t$ be the superintervals obtained through the above process and
   $(a_1,\dots,a_t) \in \{\uparrow,\downarrow\}^t$ be the corresponding string of bits.

\item Draw $m = \Theta(k \cdot \log (n/k) / \eps^3)$ samples $\mathbf{s}' $ from $p$.
For each superinterval $S_i$, $i\in[t]$, run $A^{a_i}$ on the conditional distribution $p_{S_i}$ of $p$ using the samples in $\mathbf{s}' \cap S_i$.
Let $\widetilde{p}_{S_i}$ be the hypothesis thus obtained.

\item Output the hypothesis $h  =  \littlesum_{i=1}^t \wh{p}(S_i) \cdot \widetilde{p}_{S_i} +
\littlesum_j \wh{p} (\{b_j\}) \cdot {\bm 1}_{b_j}.$

\end{enumerate}

\end{minipage}}

\bigskip

\noindent We are now ready to prove Theorem~\ref{thm:learnkmodal}.

\bigskip

\begin{proofof}{Theorem~\ref{thm:learnkmodal}}
Before entering into the proof we record two observations;
we state them explicitly here for the sake of the exposition.

\begin{fact} \label{fact:extreme}
Let $R \subseteq [n].$ If $p_R$ is neither non-increasing nor non-decreasing,
then $R$ contains at least one left extreme point.
\end{fact}

\begin{fact} \label{fact:testie}
Suppose that $R \subseteq [n]$ does not contain a left extreme point.
For any $\eps,\tau$, if $\TND(\eps,\tau)$ and $\TNI(\eps,\tau)$ are both run on $p_R$, then the probability that both calls return ``no'' is at most $\tau.$
\end{fact}

\begin{proof}  By Fact~\ref{fact:extreme} $p_R$ is either non-decreasing or non-increasing. If $p_R$ is non-decreasing then $\TND$ will output ``no'' with probability at most $\tau$, and similarly, if $p_R$ is non-increasing then $\TNI$ will output ``no'' with probability at most $\tau.$
\end{proof}

Since $r = \Theta(1/\tau^2)$ samples are drawn in the first step,
\new{Fact~\ref{fact:vc} (applied for $d=1$)}  implies that with probability of failure at most $1/100$
each interval $I \subseteq [n]$ has $| \wh{p}(I) - p(I) | \leq 2\tau$. For the
rest of the proof we condition on this good event.

Since every atomic interval $I \in \mathcal{I}$ has $\wh{p}(I) \geq
\eps/(10k)$ (except potentially the rightmost one), it follows that the
number $\ell$ of atomic intervals constructed in Step~2 satisfies $\ell \leq
10 \cdot (k/\eps)$. Moreover, \new{by our conditioning}, each atomic interval $I_i$
has $p(I_i) \geq 8\eps/(100k).$

Note that in Case (1) of Step~4, if $\widehat{p}[a_{j_1},b_{j_1}] \geq 2\eps/(10k)$
then it must be the case that $\widehat{p}(b_{j_1}) \geq \eps/(10k)$ (and thus
$p(b_{j_1}) \geq 8\eps/(100k)$).  In this case, by definition of how the interval $I_{j_1}$ was formed,
we must have that $I'_{j_1}= [a_{j_1},b_{j_1}-1]$ satisfies $\widehat{p}(I'_{j_1}) < \eps/(10k)$.
So both in Case 1 and Case 2, we now have that $\widehat{p}(I'_{j_1}) \leq 2 \eps/(10k)$, and thus
$p(I'_{j_1}) \leq 22\eps/(100k)$.  Entirely similar reasoning shows that
every negligible interval constructed in Steps~4 and~5 has mass at most $22\eps/(100k)$ under $p$.

In Steps 4--5 we invoke the testers $\TNI$ and $\TND$ on
the conditional distributions of (unions of contiguous) atomic intervals.
Note that we need enough samples in every atomic interval, since otherwise
the testers provide no guarantees. We claim that with probability at least
$99/100$ over the sample $\mathbf{s}$ of Step~3, {\em each} atomic interval gets
$b = \Omega \left( \new{(k/\eps^2)} \cdot \log(1/\tau') \right)$ samples.
This follows by a standard coupon collector's argument, which we now provide.
As argued above, each atomic interval has probability mass $\Omega(\eps/k)$ under $p$.
So, we have $\ell = O(k/\eps)$ bins (atomic intervals), and we want each bin
to contain $b$ balls (samples). It is well-known~\cite{NS:60} that after taking
$\Theta(\ell \cdot \log \ell + \ell \cdot b \cdot \log \log \ell)$ samples from $p$,
with probability $99/100$ each bin will contain the desired number of balls.
The claim now follows by our choice of parameters. Conditioning on this event,
any execution of the testers $\TND(\eps,\tau')$ and $\TNI(\eps,\tau')$ 
in Steps~4 and 5 will have the guaranteed completeness and soundness properties.

In the execution of Steps~4 and~5, there are a total of at most
$\ell$ occasions when $\TND(\eps,\tau')$ and $\TNI(\eps,\tau')$
are both run over some union of contiguous atomic intervals.
By Fact~\ref{fact:testie} and a union bound,
the probability that (in any of these instances the interval does not
contain a left extreme point and yet both calls return ``no'') is at
most $(10k/\eps)\tau' \leq 1/200.$  So with failure probability at most $1/200$ for this
step, each time Step~4 identifies a group of
consecutive intervals $I_{j},\dots,I_{j+r}$ such that both $\TND$ and
$\TNI$ output ``no'', there is a left extreme point in $\cup_{i=j}^{j+r} I_i.$
Since $p$ is $k$-modal, it follows that with failure probability at
most $1/200$ there are at most $k+1$ total repetitions of
Step~4, and hence the number $t$ of superintervals obtained is at most $k+1.$

We moreover claim that with very high probability each of the $t$
superintervals $S_i$ is very close to non-increasing or non-decreasing (with
its correct orientation given by $a_i$):
\begin{claim} \label{claim:Siisright}
With failure probability at most $1/100$, each $i \in [t]$ satisfies the
following:  if $a_i=\uparrow$ then $p_{S_i}$ is $\eps$-close to a non-decreasing
distribution and if $a_i=\downarrow$ then $p_{S_i}$ is $\eps$-close to a non-increasing
distribution.
\end{claim}

\begin{proof}
There are at most $2\ell \leq 20k/\eps$ instances when either
$\TNI$ or $\TND$ is run on a union of contiguous intervals.
For any fixed execution of $\TNI$ over an interval $I$,
the probability that $\TNI$ outputs ``yes''
while $p_I$ is $\eps$-far from every non-increasing distribution over $I$
is at most $\tau'$, and similarly for $\TND.$  A union bound
and the choice of $\tau'$ conclude the proof of the claim.
\end{proof}

Thus we have established that with overall failure probability at most
$5/100$, after Step~5 the interval $[n]$ has been partitioned into:
\begin{enumerate}
\item A set $\{S_i\}_{i=1}^t$ of $t \leq k+1$ superintervals, with $p(S_i) \geq 8\eps/(100k)$
and $p_{S_i}$ being $\eps$-close to either non-increasing or non-decreasing according to the value of bit $a_i$.

\item  A set  $\{I'_i\}_{i=1}^{t'}$ of $t' \leq k+1$ negligible intervals, such that $p(I'_i) \leq 22 \eps/(100k)$.

\item A set $\{b_i\}_{i=1}^{t''}$ of $t'' \leq k+1$ heavy points, each with $p(b_i) \geq 8 \eps/(100k).$
\end{enumerate}
\noindent We condition on the above good events, and bound from above the expected total variation distance (over the sample $\mathbf{s}'$). In particular, we have the following lemma:

\begin{lemma} \label{lemma:exp-dtv}
{Conditioned on the above good events 1--3,} we have that $\E_{\mathbf{s}'}\left[\dtv(h,p)\right] {=} O(\eps).$
\end{lemma}

\begin{proofof}{Lemma~\ref{lemma:exp-dtv}}
By the discussion preceding the lemma statement,
the domain $[n]$ has been partitioned into a set of superintervals, a set of negligible intervals and a set of heavy points. As a consequence, we can write
$$p  =  \littlesum_{j=1}^t p(S_j) \cdot p_{S_j} + \littlesum_{j=1}^{t''} p (\{b_j\}) \cdot {\bm 1}_{b_j} + \littlesum_{j=1}^{t'} p (I'_j) \cdot p_{I'_j} .$$
Therefore, we can bound the total variation distance as follows:
\begin{eqnarray*}
\dtv(h,p) \leq  \littlesum_{j=1}^{t} |  \wh{p}(S_j)  - p(S_j) | +
\littlesum_{j=1}^{t''} |  \wh{p}(b_j) - p(b_j)| \\ +
\littlesum_{j=1}^{t'} p (I'_j) +
\littlesum_{j=1}^{t} p(S_j) \cdot \dtv (\widetilde{p}_{S_j}, p_{S_j}).
\end{eqnarray*}
Recall that each term in the first two sums is bounded from above by $2\tau$. Hence, the contribution of these terms to the RHS is at most $2\tau\cdot (2k+2) \leq \eps/10$. Since each negligible interval $I'_j$ has $p(I'_j) \leq 22\eps/(100k)$, the contribution of the third sum is at most $t' \cdot 22 \eps/(100k) \leq \eps/4$. It thus remains to bound the contribution of the last sum.

We will show that
$$
\E_{\mathbf{s}'} \left[  \littlesum_{j=1}^{t} p(S_j) \cdot \dtv (\widetilde{p}_{S_j}, p_{S_j}) \right] {=} O(\eps).
$$

Denote $n_i = |S_i|$. Clearly,  $\sum_{i=1}^{t} n_i \leq n$.
Since we are conditioning on the good events (1)-(3), each superinterval is
$\eps$-close to monotone with a known orientation (non-increasing or
non-decreasing) given by $a_i$.  Hence we may apply Theorem~\ref{thm:birge-monotone} for each superinterval.

Recall that in Step 5 we draw a total of $m$ samples. Let $m_i$, $i \in [t]$ be the number of samples that land in $S_i$;
observe that $m_i$ is a binomially distributed random variable with
$m_i \sim \mathrm{Bin}(m, p(S_i))$.
We apply Theorem~\ref{thm:birge-monotone} for each $\eps$-monotone interval,
conditioning on the value of $m_i$, and get
$$ \dtv (\widetilde{p}_{S_i}, p_{S_i})  \leq 2 \eps +  O\left( (\log n_i / (m_i+1))^{1/3} \right).$$ Hence, we can bound from above
the desired expectation as follows
\begin{eqnarray*}
\littlesum_{j=1}^{t} p(S_j) \cdot \E_{\mathbf{s}'} \left[ \dtv (\widetilde{p}_{S_j}, p_{S_j}) \right]
\leq \left(\littlesum_{j=1}^{t} 2 \eps \cdot p(S_j) \right) + \\
O\left(\littlesum_{j=1}^{t} p(S_j) \cdot  (\log n_j)^{1/3} \cdot \E_{\mathbf{s}'} [ (m_j+1)^{-1/3}]\right).
\end{eqnarray*}
Since $\littlesum_jp(S_j) \leq 1$, to prove the lemma, it suffices to show that the second term is bounded, i.e.,
that
$$\littlesum_{j=1}^{t} p(S_j) \cdot  (\log n_j)^{1/3} \cdot \E_{\mathbf{s}'} [ (m_j+1)^{-1/3}] = O(\eps).$$
To do this, we will first need the following claim:
\begin{claim} \label{claim:binomial}
For a binomial random variable $X \sim \mathrm{Bin}(m,q)$ it holds $\E [ (X+1)^{-1/3} ] <(mq)^{-1/3}.$
\end{claim}

\begin{proof} 
Jensen's inequality 
implies that $$\E [(X+1)^{-1/3}] \leq  (\E [1/(X+1)])^{1/3}.$$
We claim that $\E [1/(X+1)] < 1/\E[X]$. This can be shown as follows:
We first recall that $\E[X] = m \cdot q$. For the expectation of the inverse,
we can write:
\begin{eqnarray*}
&& \E\left[1/(X+1)\right]  = \\
&=& \littlesum_{j=0}^m {\frac{1}{j+1} \binom{m}{j}q^j (1-q)^{m-j}} \\
&=& \frac{1}{m+1}
\cdot \littlesum_{j=0}^m {\binom{m+1}{j+1} q^j (1-q)^{m-j}}\\
& = & \frac{1}{q \cdot (m+1)} \cdot \littlesum_{i=1}^{m+1} {\binom{m+1}{i} q^i (1-q)^{m+1-i}} \\
&=& \frac{1 - (1-q)^{m+1}}{q \cdot (m+1)}  < \frac{1}{m \cdot q}.
\end{eqnarray*}
The claim now follows by the monotonicity of the mapping $x \mapsto x^{1/3}$.
\end{proof}

By Claim~\ref{claim:binomial}, applied to $m_i  \sim \mathrm{Bin}(m,p(S_i))$, we have that
$ \E_{\mathbf{s}'} [ (m_i+1)^{-1/3}] < m^{-1/3} \cdot (p(S_i))^{-1/3}.$
Therefore, our desired quantity can be bounded from above by
$$\littlesum_{j=1}^{t} \frac {p(S_j) \cdot  (\log n_j)^{1/3}} {m^{1/3} \cdot (p(S_j)) ^{1/3}} =
O(\eps) \cdot \littlesum_{j=1}^{t}{(p(S_j))^{2/3} \cdot \left(\frac{\log n_j}{ k \cdot \log(n/k)}\right)^{1/3}} .$$
We now claim that the second term in the RHS above is upper bounded by $2$.
Indeed, this follows by an application of H\"older's inequality
for the vectors $(p(S_j)^{2/3})_{j=1}^{t}$ and
$((\frac{\log n_j} { k \cdot \log(n/k)})^{1/3})_{j=1}^{t}$,
with  H\"older conjugates $3/2$ and $3$. That is,
\begin{eqnarray*}
&& \littlesum_{j=1}^{t} \left(p(S_j)\right)^{2/3} \cdot \left(\frac{\log n_j}{ k \cdot \log(n/k)}\right)^{1/3} \leq \\
 &\leq& \left( \littlesum_{j=1}^{t} p(S_j) \right)^{2/3} \cdot 
\left( \littlesum_{j=1}^{t} \frac{\log n_j}{ k \cdot \log(n/k)} \right)^{1/3} \\ 
&\leq& 2.
\end{eqnarray*}
The first inequality is  H\"older and the second uses the fact that  $\sum_{j=1}^{t} p(S_j) \leq 1$ and
$\sum_{j=1}^t \log (n_j) \leq t \cdot \log(n/t) \leq (k+1) \cdot \log(n/k)$.
This last inequality is a consequence of the concavity of the logarithm and the fact that $\sum_j n_j \leq n$. This completes the proof of the lemma.
\end{proofof}

By applying Markov's inequality and a union bound, we get that with
probability $9/10$ the algorithm {\tt Learn-kmodal}
outputs a hypothesis $h$ that has $\dtv(h,p) {=} O(\eps)$ as required.

It is clear that the algorithm has the claimed sample complexity.
The running time is also easy to analyze, as it is easy to see that
every step can be performed in polynomial time 
in the sample size. This completes the proof of Theorem~\ref{thm:learnkmodal}.
\end{proofof}

\newpage

\subsection{Testing whether a $k$-modal distribution is monotone}
\label{ssec:test}

In this section we describe and analyze the testing  algorithm $\TND$.
Given sample access to a $k$-modal distribution $q$ over $[n]$ and $\tau>0$,
our tester $\TND$ uses \new{$O(k/\tau^2)$} many samples from $q$ and has the
following properties:

\begin{itemize}
\item If $q$ is non-decreasing, $\TND$ outputs ``yes'' with probability at least $2/3$.
\vspace{-0.1cm}
\item If $q$ is $\tau$-far from non-decreasing, $\TND$ outputs ``no'' with probability at least $2/3$.
\end{itemize}
\noindent (The algorithm  $\TND(\tau, \delta)$ is obtained by repeating $\TND$ $O(\log(1/\delta))$ times and taking the majority vote.)

\medskip

\new{ 
Before we describe the algorithm we need some notation. Let $q$ be a distribution over $[n]$. 
For $a \le b < c \in [n]$ define $$E(q, a, b, c) : = \frac{q([a,b])}{(b-a+1)} - \frac{q([b+1,c])}{(c-b)}.$$
We also denote $$T(q, a, b, c) : = \frac{E(q, a, b, c)}{\frac{1}{(b-a+1)} + \frac{1}{(c-b)}}.$$
}
{Intuitively, the quantity $E(q,a,b,c)$ captures the difference
between the average value of $q$ over $[a,b]$ versus over $[b+1,c]$;
it is negative iff the average value of $q$ is higher over $[b+1,c]$
than it is over $[a,b]$.  The quantity $T(q,a,b,c)$ is a scaled
version of $E(q,a,b,c)$.}

\new{The idea behind tester $\TND$ is simple.  It is based on the
observation that if $q$ is a non-decreasing distribution,
then for any two consecutive intervals $[a,b]$ and $[b+1,c]$
the average of $q$ over $[b+1,c]$ must be at least as large
as the average of $q$ over $[a,b]$.  Thus any non-decreasing
distribution will pass a test that checks ``all'' pairs of
consecutive intervals looking for a violation. Our \newer{tester $\TND$ checks ``all'' sums of (at most) $k$ consecutive intervals looking for a violation.} Our analysis
shows that in fact such a test is complete as well as sound if the
distribution $q$ is guaranteed to be $k$-modal.  The key ingredient
is a structural result (Lemma~\ref{lem:tau-far-structure} below), which is proved using a procedure
reminiscent of ``Myerson ironing''~\cite{Myerson:81} to convert a $k$-modal
distribution to a non-decreasing distribution.}

\bigskip

\hskip-.2in \framebox{
\medskip \noindent \begin{minipage}{16.5cm}

{\tt Tester  $\TND(\tau)$}

\noindent {\bf Inputs:} $\tau > 0$; sample access to $k$-modal
distribution $q$ over $[n]$

\begin{enumerate}

\item \new{Draw $r = \Theta(k/\tau^2)$} samples
$\mathbf{s}$ from $q$ and let $\wh{q}$ be the resulting
empirical distribution.


\item If there exists \new{$\ell \in [k]$ and} \new{$\{a_i, b_i, c_i\}_{i=1}^{\ell} \in  \mathbf{s} \cup \{n\}$ with $a_i \le b_i < c_i < a_{i+1}$, $i \in [\ell-1]$,}
such that
\new{
\begin{eqnarray} \label{eqn:estimator}
\littlesum_{i=1}^{\ell} T(\wh{q}, a_i, b_i, c_i \newer{-1}) \ge \tau/4
\end{eqnarray}
}
then output ``no'', otherwise output ``yes''.
\end{enumerate}

\ignore{
\rnote{Setting up the test this way, where step 2 deals only with points
in the sample $\cup \{1,n\}$, makes the runtime analysis easier
but seems to make the analysis a little hairier, see the long rnote below.
Another option would be to have the test talk about all triples
$a\leq b < c \in [1,n]$ and then in the runtime analysis just say
that it's easy to see that if there exist any such 3 points $a,b,c$,
then there must exist 3 such points all of which are either in
$\mathbf{s} \cup \{1,n\}$ or are neighbors of those points.}
}

\end{minipage}}

\bigskip

The following theorem establishes correctness of the tester.

\begin{theorem} \label{thm: tester-correctness}
The algorithm $\TND$ uses \new{$O(k/\tau^2)$} samples from $q$, performs $\poly(k/\tau) \cdot \log n$ bit operations and satisfies the desired completeness and soundness properties.
\end{theorem}

\begin{proof}
We start by showing that the algorithm has the claimed completeness and soundness
properties.
Let us say that the sample $\mathbf{s}$ is \emph{good} \new{if 
for every collection $\mathcal{I}$ of (at most) $3k$ intervals in $[n]$ it holds 
$$\littlesum_{I \in \mathcal{I}} |q(I) - \wh{q}(I) | \le \tau/20.$$ By Fact~\ref{fact:vc}
with probability at least $2/3$ the sample $\mathbf{s}$ is good.
We henceforth condition on this event.
}  

\new{ 

For $a \leq b  < c \in [n]$ let us denote $\gamma = \left| q([a,b]) - \wh{q}([a,b])\right|$ and $\gamma' = \left| q([b+1,c]) - \wh{q}([b+1,c])\right|$. 
Then we can write
$$
|E(q, a,b,c) - E(\wh{q}, a,b,c)| \le \frac{\gamma}{b-a+1} +  \frac{\gamma'}{c-b} \le (\gamma+\gamma') \cdot \left( \frac{1}{b-a+1} +  \frac{1}{c-b} \right)
$$
which implies that
\begin{equation} \label{eqn:single-term}
|T(q, a,b,c) - T(\wh{q}, a,b,c)| \le \gamma + \gamma'.
\end{equation}
}
\new{Now consider any $\{a_i, b_i, c_i\}_{i=1}^{\ell} \in  [n]$, for some $\ell \le k$, with $a_i \le b_i < c_i < a_{i+1}$, $i \in [\ell-1]$. 
Similarly denote $\gamma_i = \left| q([a_i,b_i]) - \wh{q}([a_i,b_i])\right|$ and $\gamma'_i = \left| q([b_i+1,c_i]) - \wh{q}([b_i+1,c_i])\right|$. 
With this notation we have
\[ 
\left|  \littlesum_{i=1}^{\ell} T(q, a_i, b_i, c_i)  -  \littlesum_{i=1}^{\ell} T(\wh{q}, a_i, b_i, c_i) \right| 
\le  \littlesum_{i=1}^{\ell} |T(q, a_i,b_i,c_i) - T(\wh{q}, a_i,b_i,c_i)| \le \littlesum_{i=1}^{\ell} (\gamma_i +\gamma'_i)
\]
where we used the triangle inequality and (\ref{eqn:single-term}). 
Note that the rightmost term is the sum of the ``additive errors'' for the collection $\{ [a_i, b_i], [b_i+1, c_i] \}_{i=1}^{\ell}$ of $2\ell$ intervals.
Hence, it follows from our conditioning that the last term is bounded from above by $\tau/20$, i.e., 
\begin{equation} \label{eqn:all-terms}
\left|  \littlesum_{i=1}^{\ell} T(q, a_i, b_i, c_i)  -  \littlesum_{i=1}^{\ell} T(\wh{q}, a_i, b_i, c_i) \right| 
\le \tau/20.
\end{equation}
}

\smallskip
We first establish completeness. Suppose that $q$ is non-decreasing. Then the average probability
value in any interval $[a,b]$ is a non-decreasing function of $a$.
That is, {\em for all} $a \leq b < c \in [n]$ it holds
$E(q, a,b,c) \leq 0$, \new{hence $T(q,a,b,c) \leq 0$. This implies that for any choice of $\{a_i, b_i, c_i\}_{i=1}^{\ell} \in  [n]$ with $a_i \le b_i < c_i < a_{i+1}$, we will have 
$\littlesum_{i=1}^{\ell} T(q, a_i, b_i, c_i) \le 0$. By (\ref{eqn:all-terms}) we now get that $$\littlesum_{i=1}^{\ell} T(\wh{q}, a_i, b_i, c_i) \le \tau/20,$$ i.e., the tester 
says ``yes'' with probability at least $2/3$.
} 




\smallskip

To prove soundness, we will crucially need the following structural lemma:
\new{
\begin{lemma} \label{lem:tau-far-structure}
Let $q$ be a $k$-modal distribution over $[n]$ that is $\tau$-far from being non-decreasing. Then {\em there exists} 
$\ell \in [k]$ and $\{a_i, b_i, c_i\}_{i=1}^{\ell} \subseteq [n]^{3\ell}$ with $a_i \le b_i < c_i < a_{i+1}$, $i \in [\ell-1]$, such that 
\begin{equation} \label{eqn:str}
\littlesum_{i=1}^{\ell} T(q, a_i, b_i, c_i) \ge \tau/2.
\end{equation}

\end{lemma}
}

We first show how the soundness  follows from the lemma.
\new{
Let $q$ be a $k$-modal distribution over $[n]$ that is $\tau$-far from non-decreasing.
Denote $\mathbf{s}' := \mathbf{s} \cup \{n\} = \{s_1, s_2, \ldots, s_{r'}\}$ with $r' \le r+1$ and 
$s_j < s_{j+1}.$ We want to show that there exist points in $\mathbf{s}'$ that satisfy (\ref{eqn:estimator}). Namely, that
there exists $\ell \in [k]$ and $\{ s_{a_i}, s_{b_i}, s_{c_i} \}_{i=1}^{\ell} \in \mathbf{s}'$ with $s_{a_i} \le s_{b_i} < s_{c_i} < s_{a_{i+1}}$, $i \in [\ell-1]$, 
such that \newer{
\begin{equation} \label{eqn:final-goal2}
\littlesum_{i=1}^{\ell} T(\wh{q}, s_{a_i}, s_{b_i}, s_{c_i} \newer{-1}) \ge \tau/4.
\end{equation}
}
By Lemma~\ref{lem:tau-far-structure}, there exists
$\ell \in [k]$ and $\{a_i, b_i, c_i\}_{i=1}^{\ell} \in [n]$ with $a_i \le b_i < c_i < a_{i+1}$, $i \in [\ell-1]$, 
such that $\littlesum_{i=1}^{\ell} T(q, a_i, b_i, c_i) \ge \tau/2.$  
\newer{Combined with (\ref{eqn:all-terms}) the latter {inequality} 
implies that
\begin{equation} \label{eqn:sound}
\littlesum_{i=1}^{\ell} T(\wh{q}, a_i, b_i, c_i) \ge \tau/2 - \tau/20 > \tau/4.
\end{equation} 
First note that it is no loss of generality to assume that $\wh{q}([a_i, b_i])>0$ for all $i \in [\ell]$.
(If there is some $j \in [\ell]$ with $\wh{q}([a_j, b_j])=0$, then by definition we have $T(\wh{q}, a_j, b_j, c_j) \le 0$; hence, 
we can remove this term from the above sum and the RHS does not decrease.)

Given the domain points $\{a_i, b_i, c_i\}_{i=1}^{\ell}$ we define the sample points $s_{a_i}, s_{b_i}, s_{c_i}$ such that:
\begin{enumerate}
\item[(i)] $[s_{a_i}, s_{b_i}] \subseteq [a_i, b_i]$, 
\item[(ii)] $[s_{b_i}+1, s_{c_i}-1] \supseteq [b_i+1, c_i]$, 
\item[(iii)] $\wh{q}([s_{a_i}, s_{b_i}]) =  \wh{q}([a_i, b_i])$ and 
\item[(iv)] $\wh{q}([s_{b_i}+1, s_{c_i}-1]) =  \wh{q}([b_i+1, c_i])$.
\end{enumerate}
To achieve these properties we select:
\begin{itemize}
\item $s_{a_i}$ to be the leftmost point of the sample in $[a_i, b_i]$; $s_{b_i}$ to be the rightmost
point of the sample in $[a_i, b_i]$. Note that by our assumption that $\wh{q}([a_i, b_i])>0$ at least one sample falls in $[a_i, b_i]$. 

\item $s_{c_i}$ to be the leftmost point of the sample in $[c_i+1, n]$; or the point $n$ if $[c_i+1, n]$ has no samples or is empty.
\end{itemize}
We can rewrite (\ref{eqn:sound}) as follows:
\begin{equation} \label{eqn:sound2}
\sum_{i=1}^{\ell} \frac{\wh{q}([a_i, b_i])}{1+ \frac{b_i-a_i+1}{c_i-b_i}} \ge \tau/4 + \sum_{i=1}^{\ell} \frac{\wh{q}([b_i+1, c_i])}{1+ \frac{c_i-b_i}{b_i-a_i+1}}.
\end{equation}
Now note that by properties (i) and (ii) above it follows that $b_i-a_i+1\ge s_{b_i} - s_{a_i}+1$ and $c_i-b_i \le s_{c_i} - s_{b_i}-1$. Combining with properties
(iii) and (iv) we get 
\begin{equation} \label{eqn:useful-ena}
\frac{\wh{q}([a_i, b_i])}{1+ \frac{b_i-a_i+1}{c_i-b_i}} = \frac{\wh{q}([s_{a_i}, s_{b_i}])}{1+ \frac{b_i-a_i+1}{c_i-b_i}} \le  
\frac{\wh{q}([s_{a_i}, s_{b_i}])}{1+ \frac{s_{b_i}-s_{a_i}+1}{s_{c_i}-s_{b_i}-1}}
\end{equation}
and similarly
\begin{equation} \label{eqn:useful-dyo}
\frac{\wh{q}([b_i+1, c_i])}{1+ \frac{c_i-b_i}{b_i-a_i+1}} = \frac{\wh{q}([s_{b_i}+1, s_{c_i}-1])}{1+ \frac{c_i-b_i}{b_i-a_i+1}} \ge  
\frac{\wh{q}([s_{b_i}+1, s_{c_i}-1])}{1+ \frac{s_{c_i}-s_{b_i}-1}{s_{b_i}-s_{a_i}+1}}.
\end{equation}
A combination of (\ref{eqn:sound2}), (\ref{eqn:useful-ena}), (\ref{eqn:useful-dyo}) yields the
desired result (\ref{eqn:final-goal2}).

}

}

\medskip

\noindent It thus remains to prove Lemma~\ref{lem:tau-far-structure}.

\medskip

\begin{proofof}{Lemma~\ref{lem:tau-far-structure}}
We will prove the contrapositive. \new{ Let $q$ be a $k$-modal distribution over $[n]$ such that
for any $\ell \le k$ and $\{a_i, b_i, c_i\}_{i=1}^{\ell} \subseteq [n]^{3\ell}$ such that $a_i \le b_i < c_i < a_{i+1}$, $i \in [\ell-1]$,
we have
\begin{equation} \label{eqn:contrapositive}
\littlesum_{i=1}^{\ell} T(q, a_i, b_i, c_i) \le \tau/2.
\end{equation}
}
We will construct a non-decreasing distribution $\widetilde{q}$ that is $\tau$-close to $q$. 

{The high level idea of the argument is as follows:}
the construction of $\widetilde{q}$ proceeds in \new{(at most)} $k$ stages
where in each stage, we reduce the number of modes by at least one and incur 
small error in the total variation distance.
In particular, we iteratively construct a sequence of distributions $\{q^{(i)} \}_{i=0}^{\new{\ell}}$, $q^{(0)} = q$ and $q^{(\new{\ell})} = \widetilde{q}$,
for some $\new{\ell \le k}$, such that for all $i \in [\new{\ell}]$ we have that $q^{(i)}$ is $(k-i)$-modal and $\dtv(q^{(i-1)}, q^{(i)}) \leq \new{2\tau_i}$,
{where the quantities $\tau_i$ will be defined in the course
of the analysis below.}
\new{By appropriately using (\ref{eqn:contrapositive}), we will show that 
\begin{equation} \label{eqn:final-goal}
\littlesum_{i=1}^{\new{\ell}} \tau_i \le \tau/2.
\end{equation} 
Assuming this, it follows from 
the triangle inequality that
$$\dtv(\widetilde{q}, q ) \le \littlesum_{i=1}^{\new{\ell}} \dtv(q^{(i)}, q^{(i-1)}) \le 2 \cdot \littlesum_{i=1}^{\new{\ell}} \tau_i  \le \tau$$
as desired, where the last inequality uses (\ref{eqn:final-goal}).
}

\medskip

Consider the graph (histogram) of the discrete density $q$. The $x$-axis represents the $n$ points of the domain and the $y$-axis
the corresponding probabilities. We first informally describe how to obtain $q^{(1)}$ from $q$. The construction of \new{$q^{(i)}$ from $q^{(i-1)}$, $i \in [\ell]$,}
is \new{essentially} identical. Let $j_{\new{1}}$ be the {\em leftmost} \new{(i.e., having minimum $x$-coordinate)} 
left-extreme point (mode) of $q$, and assume that it is a local maximum with height (probability mass) $q(j_{\new{1}})$.
(A symmetric argument works for the case that it is a local minimum.)
The idea of the proof is based on the following simple process (reminiscent of Myerson's ironing process~\cite{Myerson:81}): 
We start with the horizontal line $y = q(j_{\new{1}})$ and move it downwards
until we reach a height $\new{h_1} < q(j_{\new{1}})$ so that the total mass ``cut-off'' equals the mass
``missing'' to the right; then we make the distribution ``flat'' in the corresponding interval (hence, reducing the number of modes by at least one).

We now proceed with the formal argument, assuming as above that
the leftmost left-extreme point $j_{\new{1}}$ of $q$ is a local maximum.
We say that the line $y=h$ {\em intersects} a point $i \in [n]$ in the domain of $q$ if
$q(i) \geq h$.  The line $y = h$, $h \in [0,q(j_{\new{1}})]$, intersects the graph of $q$ at a unique interval $I(h) \subseteq [n]$ that contains $j_{\new{1}}$.
Suppose $I(h) = [a(h),b(h)]$, where $a(h),b(h) \in [n]$ depend on $h$.
By definition this means that $q(a(h)) \geq h$ and $q(a(h)-1) < h$
{(since $q$ is supported on $[n],$ 
we adopt the convention that $q(0)=0$).}
Recall that the distribution $q$ is non-decreasing
in the interval $[1,j_{\new{1}}]$ and that $j_{\new{1}} \geq a(h)$.
The term ``the mass cut-off by the line $y=h$'' means
the quantity $$A(h) = q\left( I(h) \right) - h \cdot (b(h)-a(h)+1),$$
i.e., the ``mass of the interval $I(h)$ above the line.''

The height $h$ of the line $y=h$ defines the points $a(h), b(h)\in [n]$ as described above.
We consider values of $h$ such that $q$ is unimodal (increasing then decreasing) over $I(h)$.
In particular, let $j'_{\new{1}}$ be the leftmost mode of $q$ to the right of $j_{\new{1}}$,
i.e., $j'_{\new{1}}>j_{\new{1}}$ and $j'_{\new{1}}$ is a local minimum.
We consider values of $h \in (q(j'_{\new{1}}), q(j_{\new{1}}))$. For such values, the
interval $I(h)$ is indeed unimodal (as $b(h) < j'_{\new{1}}$).
For  $h \in (q(j'_{\new{1}}), q(j_{\new{1}}))$ we define the point $c(h) \geq j'_{\new{1}}$ as follows:
It is the rightmost point of the largest interval containing
$j'_{\new{1}}$ whose probability mass does not exceed $h$. That is, all points
in $[j'_{\new{1}}, c(h)]$ have probability mass at most $h$ and $q(c(h)+1) > h$
(or $c(h)=n$).

Consider the interval $J(h) = [b(h)+1, c(h)]$.
This interval is non-empty, since $b(h)< j'_{\new{1}} \leq c(h)$.
(Note that $J(h)$ is not necessarily a unimodal interval; it contains at
least one mode $j'_{\new{1}}$ \new{of $q$}, but it may also contain more modes.)
The term ``the mass missing to the right of the line $y=h$''
means the quantity $$B(h) = h \cdot (c(h)-b(h)) - q\left(J(h)\right).$$

Consider the function $C(h) = A(h) - B(h)$ over $[q(j'_{\new{1}}), q(j_{\new{1}})]$.
This function is continuous in its domain; moreover, we have that
$$C\left(q(j_{\new{1}})\right) = A\left(q(j_{\new{1}})\right) - B\left(q(j_{\new{1}})\right) <0,$$
as $A\left(q(j_{\new{1}})\right)=0$,
and $$C\left(q(j'_{\new{1}})\right) = A\left(q(j'_{\new{1}})\right) - B\left(q(j'_{\new{1}})\right) > 0,$$
as $B\left(q(j'_{\new{1}})\right)=0$. Therefore, by the intermediate value theorem,
there exists a value $h_{\new{1}} \in (q(j'_{\new{1}}), q(j_{\new{1}}))$ such that $$A(h_{\new{1}}) = B(h_{\new{1}}).$$

The distribution $q^{(1)}$ is constructed as follows: We move the
mass $\tau_1 = A(h_{\new{1}})$ from $I(h_{\new{1}})$ to $J(h_{\new{1}})$.
\new{Note that the distribution $q^{(1)}$ is identical to $q$ outside the interval $[a(h_1), c(h_1)]$, hence the leftmost
mode of $q^{(1)}$ is in $(c(h_{\new{1}}) , n]$. 
It is also clear that $$\dtv(q^{(1)}, q) \leq 2\tau_1.$$} 

\new{Let us denote $a_1= a(h_1)$, $b_1= b(h_1)$ and $c_1= c(h_1)$.}
We claim that $q^{(1)}$ has at least one mode less than $q$. Indeed,
$q^{(1)}$ is non-decreasing in $[1, a_1-1]$ and constant in $[a_1,c_1]$.
(By our ``flattening'' process, all the points in the latter interval have probability mass exactly $h_{\new{1}}$.)
Recalling that \new{$$q^{(1)}(a_1) = h_{\new{1}} \geq q^{(1)}(a_1-1) = q(a_1-1),$$}
we deduce that $q^{(1)}$ is non-decreasing in $[1,\new{c_1}]$.

We will now argue that \new{
\begin{equation} \label{eqn:tau1}
\tau_1 =  T(q, a_1, b_1, c_1).
\end{equation}
}
Recall that we have $A(h_{\new{1}})  = B(h_{\new{1}}) = \new{\tau_1}$, which can be written as
\new{
\begin{eqnarray*}
q([a_1, b_1])  - h_1 \cdot (b_1-a_1+1)  = h_1 \cdot (c_1-b_1)  - q([b_1+1,c_1]) = \tau_1 .
\end{eqnarray*}
}
From this, we get
\new{
\begin{eqnarray*}
\frac{q([a_1,b_1])}{(b_1-a_1+1)}  - \frac{q([b_1+1,c_1])}{(c_1-b_1)}
= \frac{\tau_1}{(b_1-a_1+1)} + \frac{\tau_1}{(c_1-b_1)}
\end{eqnarray*}
or equivalently
\begin{eqnarray*}
E\left(q, a_1, b_1, c_1 \right) = \frac{\tau_1}{(b_1-a_1+1)} + \frac{\tau_1}{(c_1-b_1)}
\end{eqnarray*}
which gives (\ref{eqn:tau1}).
}

\new{ We construct $q^{(2)}$ from $q^{(1)}$ using the same procedure. 
Recalling that the leftmost mode  of $q^{(1)}$ lies in the interval $(c_1 , n]$ 
an identical argument as above implies that
$$\dtv(q^{(2)}, q^{(1)}) \leq 2\tau_2$$
where 
$$\tau_2 =  T(q^{(1)}, a_2, b_2, c_2)$$ 
for some $a_2, b_2, c_2\in [n]$ satisfying $c_1 < a_2 \le b_2 < c_2$. Since $q^{(1)}$
is identical to $q$ in $(c_1, n]$, it follows that 
$$\tau_2 =  T(q, a_2, b_2, c_2).$$ 
\new{We continue this process iteratively for $\ell \le k$ stages until we obtain a non-decreasing distribution $q^{(\ell)}$.
(Note that we remove at least one mode in each iteration, hence it may be the case that $\ell < k$.)}
It follows inductively that for all $i \in [\ell]$, we have that $\dtv(q^{(i)}, q^{(i-1)}) \leq 2\tau_i$
where $\tau_i =  T(q, a_i, b_i, c_i)$, for $c_{i-1} < a_{i} \le b_{i} < c_{i}$.

We therefore conclude that $$\littlesum_{i=1}^{\ell} \tau_i = \littlesum_{i=1}^{\ell} T(q, a_i, b_i, c_i)$$ which is bounded from above by 
$\tau/2$ by (\ref{eqn:contrapositive}).
}
This establishes (\ref{eqn:final-goal}) 
completing the proof of Lemma~\ref{lem:tau-far-structure}.
\end{proofof}

The upper bound on the sample complexity of the algorithm is straightforward, since
only Step~1 uses samples. 

\newer{It remains to analyze the running time. The only non-trivial computation is in Step~2 where we need to decide whether 
there exist $\ell \le k$ ``ordered triples'' $\{a_i, b_i, c_i\}_{i=1}^{\ell} \in  \mathbf{s}'$ with $a_i \le b_i < c_i < a_{i+1}$, $i \in [\ell-1]$,
such that $\littlesum_{i=1}^{\ell} T(\wh{q}, a_i, b_i, c_i-1) \ge \tau/4.$
Even though a naive brute-force implementation would need time $\Omega(r^k) \cdot \log n$, 
there is a simple dynamic programming algorithm that runs in $\poly(r, k) \cdot \log n$
time.

We now provide the details. Consider the objective function
$$\mathcal{T}(\ell) = \max \left\{ \littlesum_{i=1}^{\ell} T(\wh{q}, a_i, b_i, c_i-1) \mid 
\{a_i, b_i, c_i\}_{i=1}^{\ell} \in \mathbf{s}'  \textrm{ with } a_i \le b_i < c_i < a_{i+1}, i \in [\ell-1] \right\}, $$
for $\ell \in [k].$ We want to decide whether $\max_{\ell \le k}\mathcal{T}(\ell) \ge \tau/4$. 
For $\ell \in [k]$ and $j \in [r']$, we use dynamic programming to compute the quantities 
$$\mathcal{T}(\ell, j) = \max \left\{ \littlesum_{i=1}^{\ell} T(\wh{q}, a_i, b_i, c_i-1) \mid 
\{a_i, b_i, c_i\}_{i=1}^{\ell} \in \mathbf{s}'  \textrm{ with } a_i \le b_i < c_i < a_{i+1}, i \in [\ell-1] \textrm{ and } c_{\ell} = s_j \right\}.$$
(This clearly suffices as $\mathcal{T}(\ell) = \max_{j \in [r']} \mathcal{T}(\ell, j)$.)
The dynamic program is based on the recursive identity 
$$\mathcal{T}(\ell+1, j) = \max_{j' \in [r'],  j'<j}\mathcal{T}(\ell, j') +  \mathcal{T}'(j'+1, j).$$
where we define 
$\mathcal{T}'(\alpha,\beta)  =  \max \{ T(\wh{q}, a, b, \beta) \mid  a,b \in \mathbf{s}', 
\alpha \leq a \leq b < \beta \}.$
Note that all the values $\mathcal{T}'(j'+1,j)$ (where $j',j \in [r']$ and $j'<j$) can be computed in $O(r^3)$ time. Fix $\ell \in [k]$. Suppose we have computed all
the values $\mathcal{T}(\ell, j')$, $j' \in [r']$. Then, for fixed $j \in [r']$, we can compute the value $\mathcal{T}(\ell+1, j)$ in time $O(r)$ 
using the above recursion. Hence, the total running time of the algorithm is $O(kr^2+r^3)$.
}
\noindent This completes the run time analysis and the proof of Theorem~\ref{thm: tester-correctness}.
\end{proof}

\ignore{

\section{Lower Bound} \label{sec:lb}

\rnote{Do we want this section at all, or just make the assertion in the
introduction (Proposition 1) the way we do now?}

A $\Omega(k \log(n/k)/\eps^3)$ lower bound for learning $k$-modal distributions.

I think it's clear that in order to learn a $k$-modal distribution to accuracy $\eps$
you need $\Omega(k\log(n/k)/\eps^3)$ samples.  Suppose you are promised that the distribution's $k$
modes occur at $n/k$, $2n/k$, etc, and that the weight over each of these $k$ regions is exactly
$1/k$.  In order to learn the distribution to accuracy $\eps$ you need to learn at least
half of the $k$ induced monotone distribution learning problems to accuracy at worst $2 \eps$.
Each of these problems is over an $n/k$-size domain so you need $\Omega(\log(n/k)/\eps^3)$ samples
for each one.

\rnote{To do:  make this a little less fast and loose}

}


\section{Conclusions and future work} \label{sec:conclusions}


At the level of techniques, this work illustrates the viability of
a new general strategy for developing
efficient learning algorithms, namely by using ``inexpensive'' property testers
to decompose a complex object (for us these objects are $k$-modal
distributions) into simpler objects (for us these are monotone distributions)
that can be more easily learned.  It would be interesting to apply this
paradigm in other contexts such as learning Boolean functions.

At the level of the specific problem we consider -- learning
$k$-modal distributions -- our results show that $k$-modality is a useful
type of structure which can be strongly exploited by sample-efficient and
computationally efficient learning algorithms.  Our results motivate the
study of computationally efficient learning algorithms for distributions
that satisfy other kinds of ``shape restrictions.''  Possible directions here
include multivariate $k$-modal distributions, log-concave distributions,
monotone hazard rate distributions and more. 

At a technical level, any improvement in the sample
complexity of our property testing algorithm of
Section~\ref{ssec:test} would directly improve
the ``extraneous'' additive
\new{$\tilde{O}(k^2/\eps^3)$} term in the sample complexity
of our algorithm.  We suspect that it may be possible to improve
our testing algorithm (although we note that it is easy to give an
\new{$\Omega(\sqrt{k}/\eps^2)$} lower bound using standard constructions).

Our learning algorithm is not proper, i.e., it outputs a hypothesis that is not necessarily $k$-modal.
Obtaining an efficient proper learning algorithm is an interesting question.
Finally, it should be noted that our approach for learning $k$-modal distributions requires a priori knowledge
of the parameter $k$. We leave the case of unknown $k$ as an intriguing open problem.


\section*{Acknowledgement}
We thank the anonymous reviewers for their helpful comments.

\appendix

\section{Birg{\'e}'s algorithm as a semi-agnostic learner}
\label{sec:birge-agnostic}

In this section we briefly explain why Birg{\'e}'s algorithm~\cite{Birge:87b}
also works in the semi-agnostic setting,
{thus justifying the claims about its performance made in the
statement of Theorem \ref{thm:birge-monotone}.}
To do this, we need to explain
his approach. For this, we will need the following \new{fact (which follows as a special case of the VC inequality, {Theorem~\ref{thm:vc-inequality}}}), 
which gives
a tight bound on the number of samples required to
learn an arbitrary distribution with respect to {\em total variation distance}.

\begin{fact} \label{thm:folklore}
Let $p$ be any distribution over $[n]$. We have: $\E [ \dtv ( p, \wh{p}_m) ]  =  O( \sqrt{n/m}).$
\end{fact}

Let $p$ be a non-increasing distribution over $[n]$.
(The analysis for the non-decreasing case is identical.)
Conceptually, we view algorithm $\ANI$ as working in three steps:

\begin{itemize}
\item  In the first step, it partitions the set $[n]$ into a carefully chosen
set $I_1, \ldots, I_{\ell}$ of consecutive intervals, with $\ell = O(m^{1/3} \cdot (\log n)^{2/3})$.
Consider the {\em flattened} distribution $p_f$ over $[n]$ obtained from $p$ by averaging the weight
that $p$ assigns to each interval over the entire interval. That is, for $j \in [\ell]$ and
$i \in I_j$, $p_f(i) = \sum_{t \in I_j} p(t) / |I_j|$.
Then a simple argument given in~\cite{Birge:87b} gives
that $\dtv(p_f, p) = O\left( ({\log n / (m+1)})^{1/3} \right).$

\item Let $p_r$ be the {\em reduced} distribution corresponding to $p$ and the partition $I_1, \ldots, I_{\ell}$.
That is, $p_r$ is a distribution over $[\ell]$ with $p_r(i) = p(I_i)$ for $i \in [\ell]$. In the second step, the algorithm uses the $m$ samples to learn $p_r$. (Note that $p_r$ is not necessarily monotone.)
After $m$ samples, one obtains a hypothesis $\widehat{p_r}$ such that
$\E [\dtv(p_r, \widehat{p_r})] =  O\left( \sqrt{\ell / m} \right) = O\left( ({\log n / (m+1)})^{1/3} \right)$.  The first equality follows from \new{Fact}~\ref{thm:folklore} (since $p_r$ is distribution over $\ell$ elements) and
the second inequality follows from the choice of $\ell$.

\item Finally, the algorithm outputs the flattened hypothesis $(\widehat{p_r})_f$ over $[n]$
corresponding to  $\widehat{p_r}$, i.e., obtained by $\widehat{p_r}$ by subdividing
the mass of each interval uniformly within the interval. It follows from the above two steps that
$\E[\dtv ( (\widehat{p_r})_f, p_f)] =O\left( ({\log n / (m+1)})^{1/3} \right).$

\item The combination of the first and third steps yields that
$\E[\dtv ( (\widehat{p_r})_f, p)] =O\left( ({\log n / (m+1)})^{1/3} \right).$

\end{itemize}

\medskip

The above arguments are entirely due to Birg{\'e}~\cite{Birge:87b}.
We now explain how his analysis can be extended to show that his algorithm is
in fact a semi-agnostic learner as claimed in Theorem~\ref{thm:birge-monotone}.
To avoid clutter in the expressions below let us fix
$\delta: =O\left( ({\log n / (m+1)})^{1/3} \right)$.

The second and third steps in the algorithm description
above are used to learn the distribution $p_f$ to variation distance $\delta$.
Note that these steps do not use the assumption that $p$ is non-increasing.
The following claim, which generalizes Step 1 above, says that if $p$ is $\tau$-close to non-increasing,
the flattened distribution $p_f$ (defined as above) is $(2\tau+\delta)$-close to $p$.
Therefore, it follows that, for such a distribution $p$,
algorithm $\ANI$ succeeds with expected (total variation distance) error
$(2\tau+\delta)+\delta$.

\smallskip

\noindent We have:

\begin{claim}
Let $p$ be a distribution over $[n]$ that is $\tau$-close to non-increasing.
Then, the flattened distribution $p_f$ (obtained from $p$ by averaging its
weight on every interval $I_j$) satisfies $\dtv(p_f,p) \leq (2\tau+\delta)$.
\end{claim}

\begin{proof}
Let $p^{\downarrow}$ be the non-increasing distribution that is $\tau$-close to $p$.
Let $\tau_j$ denote the $L_1$-distance between $p$ and $p^{\downarrow}$ in the interval $I_j$.
Then, we have that
\begin{equation} \label{eqn:one}
\sum_{j=1}^{\ell} \tau_j \leq \tau.
\end{equation}

By Birg{\'e}'s arguments, it follows that the flattened distribution $(p^{\downarrow})_f$
corresponding to $p^{\downarrow}$ is $\delta$-close to $p^{\downarrow}$,
hence $(\tau+\delta)$-close to $p$. That is,
\begin{equation} \label{eqn:two}
\dtv \left( (p^{\downarrow})_f, p \right)  \leq \tau+\delta.
\end{equation}
We want to show that
\begin{equation} \label{eqn:three}
\dtv \left( (p^{\downarrow})_f, p_f \right) \leq \tau.
\end{equation}
Assuming (\ref{eqn:three}) holds, we can conclude by the triangle inequality that
$$\dtv \left(p, p_f \right) \leq 2\tau+\delta$$
as desired.

Observe that, by assumption, $p$ and $p^{\downarrow}$ have
$L_1$-distance at most $\tau_j$ in each $I_j$ interval.
In particular, this implies that, for all $j \in [\ell]$, it holds
$$ \left| p(I_j) - p^{\downarrow}(I_j) \right| \leq \tau_j.$$
Now note that, within each interval $I_j$, $p_f$ and $(p^{\downarrow})_f$ are both uniform.
Hence, the contribution of $I_j$ to the variation distance between  $p_f$ and $(p^{\downarrow})_f$
is at most $|p(I_j) - p^{\downarrow}(I_j)|$.

Therefore, by (\ref{eqn:one}) we deduce
$$\dtv (p_f , (p^{\downarrow})_f ) \leq \tau$$
which completes the proof of the claim.
\end{proof}

\section{Hypothesis Testing} \label{sec:choosehypothesis}

Our hypothesis testing routine {\tt Choose-Hypothesis}$^p$ runs a simple ``competition'' to choose a winner between two candidate hypothesis distributions $h_1$ and $h_2$  over $[n]$ that it is given in the input either explicitly, or in some succinct way.  We show that if at least one of the two candidate hypotheses is close to the target distribution $p$, then with high probability over the samples drawn from $p$ the routine selects as winner a candidate that is close to $p$.
  This basic approach of running a competition between candidate hypotheses is quite similar to the ``Scheff\'e estimate'' proposed by Devroye and Lugosi (see \cite{DL96,DL97}
and Chapter~6 of \cite{DL:01}), which in turn built closely on the work of \cite{Yatracos85}, but
there are some small differences between our approach and theirs; the \cite{DL:01} approach uses a notion of the ``competition'' between two hypotheses which is not symmetric under swapping the two competing hypotheses, whereas our competition is {symmetric}.

We now prove Theorem~\ref{thm:choosehypothesis}.

\medskip

\begin{proofof}{Theorem ~\ref{thm:choosehypothesis}}
Let ${\cal W}$ be the support of $p$. To set up the competition between $h_1$ and $h_2$, we define the following subset of ${\cal W}$:
\begin{align}
&{\cal W}_1={\cal W}_1(h_1,h_2) := \left\{w \in \mathcal{W}~\vline~h_1(w) > h_2(w) \right\}. \label{eq:W1}
\end{align}

\noindent Let then $p_1 = h_1({\cal W}_1)$ and $q_1 = h_2({\cal W}_1)$. Clearly, $p_1 > q_1$ and
 $\dtv(h_1, h_2) = p_1-q_1$.


The competition between $h_1$ and $h_2$ is carried out as follows:

\begin{enumerate}
\item  If $p_1-q_1\leq 5 \eps'$, declare a draw and return either $h_i$. Otherwise:

\item  Draw $m=O\left({\log(1/\delta') \over \eps'^2}\right)$ samples $s_1,\ldots,s_m$ from $p$, and let $\tau = {1 \over m} | \{i~|~s_i \in {\cal W}_1 \}|$ be the fraction of  samples that fall inside ${\cal W}_1.$

\item  If $\tau > p_1- {3 \over 2} \eps'$, declare $h_1$ as winner and return $h_1$; otherwise,

\item  if $\tau < q_1+ {3 \over 2} \eps'$, declare $h_2$ as winner and return $h_2$; otherwise,

\item  declare a draw and return either $h_i$.

\end{enumerate}

It is not hard to check that the outcome of the competition does not depend on the ordering of the pair of distributions provided in the input; that is, on inputs $(h_1,h_2)$ and $(h_2,h_1)$ the competition outputs the same result for a fixed sequence of samples $s_1,\ldots,s_m$ drawn from $p$.

The correctness of {\tt Choose-Hypothesis} is an immediate consequence of the following lemma.


\begin{lemma} \label{lem:kostas3}
Suppose that $\dtv(p,h_1) \leq \eps'$. Then:
\begin{itemize}

\item[(i)] If $\dtv(p, h_2)>6 \eps'$, then the probability that the competition between
$h_1$ and $h_2$ does not declare $h_1$ as the winner is at most $e^{- {m \eps'^2 /2 } }$.  (Intuitively,
if $h_2$ is very bad then it is very likely that $h_1$ will be declared winner.)

\item[(ii)]  If $\dtv(p, h_2)>4 \eps'$, the probability that the competition between $h_1$
and $h_2$ declares $h_2$ as the winner is at most $e^{- {m \eps'^2/ 2 } }$.  (Intuitively, if $h_2$
is only moderately bad then a draw is possible but it is very unlikely that $h_2$ will be declared winner.)
\end{itemize}
\end{lemma}

\begin{proof}
Let $r=p({\cal W}_1)$. The definition of the total variation distance implies that $|r-p_1| \le \eps'$. Let us define the $0/1$ (indicator) random variables $\{Z_j\}_{j=1}^m$ as $Z_j=1$ iff $s_j \in {\cal W}_1$. Clearly, $\tau={1 \over m} \sum_{j=1}^m Z_j$ and $\mathbb{E}[\tau]=\mathbb{E}[Z_j]=r$. Since the $Z_j$'s are mutually independent, it follows from the Chernoff bound that $\Pr[\tau \le r-{\eps'/2}] \le e^{- {m \eps'^2/2 } }$. Using $|r-p_1| \le \eps'$ we get that $\Pr[\tau\le p_1-{3\eps'/2}] \le e^{- {m \eps'^2/2 } }$.
\begin{itemize}
\item For part (i): If $\dtv(p, h_2) > 6 \eps'$, from the triangle inequality we get that $p_1-q_1=\dtv(h_1, h_2) > 5 \eps' $.  Hence, the algorithm will go beyond Step 1, and with probability at least $1-e^{- {m \eps'^2/2 } }$, it will stop at Step 3, declaring $h_1$ as the winner of the competition between $h_1$ and $h_2$.

\item For part (ii):  If $p_1-q_1 \leq 5 \eps'$ then the competition declares a draw, hence $h_2$ is not the winner. Otherwise we have $p_1 - q_1 > 5\eps'$ and the above arguments imply that the competition between $h_1$ and $h_2$ will declare $h_2$ as the winner with probability at most $e^{- {m \eps'^2/2 } }$.
\end{itemize}
This concludes the proof of Lemma~\ref{lem:kostas3}.
\end{proof}

The proof of the theorem is now complete.
\end{proofof}

\ignore{

We observe that the running time of {\tt Choose-Hypothesis} is linear in $m$ times the time required to
evaluate each hypothesis $H_i$ (i.e., compute the pdf) on a sample drawn from $X.$

THE ABOVE IS NOT TRUE NECESSARILY -- we need to take into account how long it takes to
compute $p_1$ and $p_2$.  A priori these are probabilities obtained by looking at $n$
points so if $H_1,H_2$ were arbitrarily complicated this could be $\Omega(n)$ time steps; for
what we get from Birge (in the sparse case) and the Binomial learner (in the binomial case), this
shouldn't be a problem.

}

\section{Using the Hypothesis Tester} \label{sec:ht}

In this section, we explain in detail how we use
the hypothesis testing algorithm {\tt Choose-Hypothesis}
throughout this paper. In particular, the algorithm
{\tt Choose-Hypothesis} is used in the following places:
\begin{itemize}

\item In Step 4 of algorithm {\tt Learn-kmodal-simple} we need
an algorithm $\ANI_{\delta'}$ (resp. $\AND_{\delta'}$) that learns
a non-increasing (resp. non-increasing) distribution within total variation distance
$\eps$ and confidence $\delta'$. Note that the corresponding algorithms
$\ANI$ and $\AND$ provided by Theorem~\ref{thm:birge-monotone} have confidence $9/10$.
To boost the confidence of $\ANI$ (resp. $\AND$)
we run the algorithm $O(\log(1/\delta'))$ times and use {\tt Choose-Hypothesis}
in an appropriate tournament procedure to select among the candidate hypothesis distributions.

\item In Step 5 of algorithm {\tt Learn-kmodal-simple} we need to
select among two candidate hypothesis distributions (with the
promise that at least one of them is close to the true
conditional distribution). In this case, we run {\tt Choose-Hypothesis} once to select between
the two candidates.

\item Also note that both algorithms {\tt Learn-kmodal-simple} and {\tt Learn-kmodal} generate
an $\eps$-accurate hypothesis with probability $9/10$. We would like to boost the probability
of success to $1-\delta$. To achieve this we again run the corresponding algorithm $O(\log(1/\delta))$
times and use {\tt Choose-Hypothesis} in an appropriate tournament
to select among the candidate hypothesis distributions.
\end{itemize}

We now formally describe the ``tournament'' algorithm to boost the confidence to $1-\delta$.

\begin{lemma} \label{lem:log-cover-size}
Let $p$ be any distribution over a finite set $\mathcal{W}$.
Suppose that ${\cal D}_\eps$ is a collection of $N$ distributions over $\mathcal{W}$ such that
there exists $q \in {\cal D}_\eps$ with $\dtv(p,q) \leq \eps$.
Then there is an algorithm that uses $O(\eps^{-2}\log N \log(1/\delta))$ samples from
$p$ and with probability $1-\delta$ outputs a distribution $p' \in {\cal D}_{\eps}$ that satisfies
$\dtv(p,p') \leq 6\eps.$
\end{lemma}

Devroye and Lugosi (Chapter 7 of \cite{DL:01}) prove a similar result by having all pairs
of distributions in the cover compete against each other using their notion of a competition, but again there
are some small differences:  their approach chooses a distribution in the cover which wins the
maximum number of competitions, whereas our algorithm chooses
a distribution that is never defeated (i.e., won or achieved a draw against all other distributions
in the cover).  Instead we follow the approach from \cite{DDS12stoclong}.

\medskip

\begin{proof}
The algorithm performs a tournament by running the competition
{\tt Choose-Hypothesis}$^p(h_i,h_j,\eps,$ $\delta/(2N))$ for every pair of distinct distributions
$h_i,h_j$ in the collection ${\cal D}_\eps$.  It outputs a distribution $q^{\star} \in {\cal D}_\eps$
that was never a loser (i.e., won or achieved a draw in all its competitions).
If no such distribution exists in ${\cal D}_\eps$ then the algorithm outputs
``failure.''

By definition, there exists some $q \in {\cal D}_\eps$ such that $\dtv(p,q) \leq \eps.$
We first argue that with high probability this distribution $q$ never
loses a competition against any other $q' \in {\cal D}_\eps$
(so the algorithm does not output ``failure'').  Consider any $q'
\in {\cal D}_\eps$. If $\dtv(p, q') > 4 \eps$, by Lemma~\ref{lem:kostas3}(ii) the probability that $q$ loses to
$q'$ is at most $2e^{-m \eps^2 /2} = O(1/N).$ On the other hand, if $\dtv(p, q') \leq 4 \delta$, the triangle
inequality gives that $\dtv(q,q') \leq 5 \eps$ and thus $q$ draws against $q'.$  A union
bound over all $N$ distributions in ${\cal D}_{\eps}$ shows that with probability $1-\delta/2$, the distribution $q$ never loses a competition.

We next argue that with probability at least $1-\delta/2$, every distribution $q' \in {\cal D}_\eps$ that never loses has small variation distance from $p.$ Fix a distribution $q'$ such that $\dtv(q',p) > 6 \eps$; Lemma~\ref{lem:kostas3}(i) implies that $q'$ loses to $q$ with probability $1 - 2e^{-m \eps^2 /2} \geq 1 - \delta/(2N)$.  A union bound gives that with probability $1-\delta/2$, every distribution $q'$ that has $\dtv(q', p) > 6 \eps$ loses some competition.

Thus, with overall probability at least $1-\delta$, the tournament does not output ``failure'' and outputs some distribution $q^{\star}$ such that $\dtv(p, q^{\star})$ is at most $6 \eps.$
This proves the lemma.
\end{proof}

We now explain how the above lemma is used in our context: Suppose we perform $O(\log(1/\delta))$ runs of a learning algorithm that constructs an $\eps$-accurate hypothesis with probability at least $9/10.$  Then, with failure probability at most $\delta/2$, at least one of the hypotheses generated is $\eps$-close to the true distribution in variation distance. Conditioning on this good event, we have a collection of distributions with cardinality $O(\log(1/\delta))$ that satisfies the assumption of the lemma. Hence, using
$O\left((1/\eps^2) \cdot \log\log(1/\delta) \cdot \log(1/\delta)\right) $ samples we can learn to accuracy
$6\eps$ and confidence $1-\delta/2$. The overall sample complexity is $O(\log(1/\delta))$ times the sample complexity of the 
{learning algorithm run with confidence $9/10$,} 
plus this additional $O\left((1/\eps^2) \cdot \log\log(1/\delta) \cdot \log(1/\delta)\right)$ term.

In terms of running time,we make the following easily verifiable remarks: When the hypothesis testing algorithm {\tt Choose-Hypothesis}  is run on a pair of distributions that are produced by Birg{\'e}'s algorithm, its running time is polynomial in the succinct description of these distributions, i.e., in $\log^2(n) / \eps$. Similarly, when {\tt Choose-Hypothesis} is run on a pair of outputs of {\tt Learn-kmodal-simple} or {\tt Learn-kmodal}, its running time is polynomial in the succinct description of these distributions. More specifically, in the former case, the succinct description has bit complexity $O \left( k \cdot \log^2 (n) /\eps^2 \right)$ (since the output consists of $O(k/\eps)$ monotone intervals, and the conditional distribution on each interval is the output of Birg{\'e}'s algorithm for that interval). In the latter case, the succinct description has bit complexity
$O\left( k \cdot \log^2 (n) /\eps \right)$, since the algorithm {\tt Learn-kmodal} constructs  only $k$ monotone intervals. Hence, in both cases, each executation of the testing algorithm performs  $\poly(k, \log n , 1/\eps)$ bit operations. Since the tournament invokes the algorithm {\tt Choose-Hypothesis} $O(\log^2(1/\delta))$ times (for every pair of distributions in our pool of $O(\log(1/\delta))$ candidates) the upper bound on the running time follows.

\newcommand{\etalchar}[1]{$^{#1}$}


\begin{thebibliography}{FPP{\etalchar{+}}98}

\bibitem[Bir87a]{Birge:87}
L.~Birg\'e.
\newblock {Estimating a density under order restrictions: Nonasymptotic minimax
  risk}.
\newblock {\em Annals of Statistics}, 15(3):995--1012, 1987.

\bibitem[Bir87b]{Birge:87b}
L.~Birg\'e.
\newblock {On the risk of histograms for estimating decreasing densities}.
\newblock {\em Annals of Statistics}, 15(3):1013--1022, 1987.

\bibitem[Bir97]{Birge:97}
L.~Birg\'e.
\newblock Estimation of unimodal densities without smoothness assumptions.
\newblock {\em Annals of Statistics}, 25(3):970--981, 1997.

\bibitem[BKR04]{BKR:04long}
T. Batu, R. Kumar, and R. Rubinfeld.
\newblock Sublinear algorithms for testing monotone and unimodal distributions.
\newblock In {\em {Proceedings of the 36th Symposium on Theory of Computing}},
  pages 381--390, 2004.

\bibitem[CKC83]{CKC:83}
L.~Cobb, P.~Koppstein, and N.H. Chen.
\newblock Estimation and moment recursion relations for multimodal
  distributions of the exponential family.
\newblock {\em J. American Statistical Association}, 78(381):124--130, 1983.

\bibitem[CT04]{ChanTong:04}
K.S. Chan and H.~Tong.
\newblock Testing for multimodality with dependent data.
\newblock {\em Biometrika}, 91(1):113--123, 2004.

\bibitem[DDS12]{DDS12stoclong}
C.~Daskalakis, I.~Diakonikolas, and R.A. Servedio.
\newblock {Learning Poisson Binomial Distributions}.
\newblock In {\em Proceedings of the 44th Symposium on Theory of Computing},
  pages 709--728, 2012.

\bibitem[DL96a]{DL97}
L.~Devroye and G.~Lugosi.
\newblock {Nonasymptotic universal smoothing factors, kernel complexity and
  Yatracos classes}.
\newblock {\em Annals of Statistics}, 25:2626--2637, 1996.

\bibitem[DL96b]{DL96}
L.~Devroye and G.~Lugosi.
\newblock A universally acceptable smoothing factor for kernel density
  estimation.
\newblock {\em Annals of Statistics}, 24:2499--2512, 1996.

\bibitem[DL01]{DL:01}
L.~Devroye and G.~Lugosi.
\newblock {\em Combinatorial methods in density estimation}.
\newblock Springer Series in Statistics, Springer, 2001.

\bibitem[dTF90]{ToledoFernandez:90}
G.A. de~Toledo and J.M. Fernandez.
\newblock Patch-clamp measurements reveal multimodal distribution of granule
  sizes in rat mast cells.
\newblock {\em Journal of Cell Biology}, 110(4):1033--1039, 1990.

\bibitem[FPP{\etalchar{+}}98]{FPPRD:98}
F.R. Ferraro, B.~Paltrinieri, F.F. Pecci, R.T. Rood, and B.~Dorman.
\newblock Multimodal distributions along the horizontal branch.
\newblock {\em The Astrophysical Journal}, 500:311--319, 1998.

\bibitem[GGR98]{GGR98}
O.~Goldreich, S.~Goldwasser, and D.~Ron.
\newblock Property testing and its connection to learning and approximation.
\newblock {\em Journal of the ACM}, 45:653--750, 1998.

\bibitem[Gol10]{PropertyTestingICS}
O.~Goldreich, editor.
\newblock {\em Property Testing: Current Research and Surveys}.
\newblock Springer, 2010.
\newblock LNCS 6390.

\bibitem[Gol11]{Goldreich:11web}
O.~Goldreich.
\newblock {Highlights of the Bertinoro workshop on Sublinear Algorithms
  (unpublished comments)}.
\newblock Posted at http://www.wisdom.weizmann.ac.il/\~ oded/MC/072.html,
  accessed June 17, 2011, 2011.

\bibitem[Gro85]{Groeneboom:85}
P.~Groeneboom.
\newblock Estimating a monotone density.
\newblock In {\em Proc. of the Berkeley Conference in Honor of Jerzy Neyman and
  Jack Kiefer}, pages 539--555, 1985.

\bibitem[Kem91]{Kemperman:91}
J.H.B. Kemperman.
\newblock Mixtures with a limited number of modal intervals.
\newblock {\em Annals of Statistics}, 19(4):2120--2144, 1991.

\bibitem[KR00]{KR00}
M.~Kearns and D.~Ron.
\newblock Testing problems with sub-learning sample complexity.
\newblock {\em J. Comp. Sys. Sci.}, 61:428--456, 2000.

\bibitem[Mur64]{Murphy:64}
E.A. Murphy.
\newblock One cause? many causes?: The argument from the bimodal distribution.
\newblock {\em J. Chronic Diseases}, 17(4):301--324, 1964.

\bibitem[Mye81]{Myerson:81}
R.B. Myerson.
\newblock Optimal auction design.
\newblock {\em Mathematics of Operations Research}, 6:58--73, 1981.

\bibitem[NS60]{NS:60}
D.~J. Newman and L.~Shepp.
\newblock The double dixie cup problem.
\newblock {\em The American Mathematical Monthly}, 67(1):pp. 58--61, 1960.

\bibitem[Rao69]{PrakasaRao:69}
B.L.S.~Prakasa Rao.
\newblock Estimation of a unimodal density.
\newblock {\em Sankhya Ser. A}, 31:23--36, 1969.

\bibitem[Ron08]{Ron:08testlearn}
D.~Ron.
\newblock {Property Testing: A Learning Theory Perspective}.
\newblock {\em Foundations and Trends in Machine Learning}, 1(3):307--402,
  2008.

\bibitem[Ron10]{Ron:10FNTTCS}
D.~Ron.
\newblock Algorithmic and analysis techniques in property testing.
\newblock {\em Foundations and Trends in Theoretical Computer Science},
  5:73--205, 2010.

\bibitem[Weg70]{Wegman:70}
E.J. Wegman.
\newblock {Maximum likelihood estimation of a unimodal density. I. and II.}
\newblock {\em Ann. Math. Statist.}, 41:457--471, 2169--2174, 1970.

\bibitem[Yat85]{Yatracos85}
Y.~G. Yatracos.
\newblock {Rates of convergence of minimum distance estimators and Kolmogorov's
  entropy}.
\newblock {\em Annals of Statistics}, 13:768--774, 1985.

\end{thebibliography}
\end{document}